\documentclass[aps,pra,twocolumn,amsmath,amssymb,nofootinbib,showpacs,superscriptaddress]{revtex4-2}
\usepackage[english]{babel}
\usepackage{latexsym}
\usepackage{graphicx}
\usepackage{epstopdf}
\usepackage{epsfig}
\usepackage{color}
\usepackage{bm}
\usepackage{amsmath}
\usepackage{amssymb}
\usepackage{fancyhdr}
\usepackage{amsthm}
\usepackage{dcolumn}
\usepackage{float}
\usepackage{hyperref}
\usepackage{color}
\usepackage{cleveref}
\usepackage[svgnames,dvipsnames]{xcolor}
\usepackage{enumerate}
\usepackage{braket}
\usepackage{bbm}

\hypersetup{hidelinks,colorlinks=true,allcolors=DarkBlue}

\usepackage{algorithm}
\usepackage{algpseudocode}
\algrenewcommand\algorithmicrequire{\textbf{Input:}}
\algrenewcommand\algorithmicensure{\textbf{Output:}}

\usepackage[svgnames,dvipsnames]{xcolor}

\usepackage{tikz}
\usetikzlibrary{arrows,decorations.text,arrows.meta,
decorations.pathmorphing,backgrounds,positioning,fit}

\DeclareMathOperator{\Tr}{Tr}

\newtheorem{theorem}{Theorem}
\newtheorem{lemma}{Lemma}
\newtheorem{proposition}{Proposition}
\newtheorem{corollary}{Corollary}

\theoremstyle{remark}
\newtheorem{remark}{Remark}

\begin{document} 

\title{Bounds on quantum conference key agreement in pair-entangled networks}
\author{Justus Neumann}
\author{Hermann Kampermann}
\author{Dagmar Bruß}

\affiliation{Heinrich Heine University D\"{u}sseldorf,
Faculty of Mathematics and Natural Sciences, Institute for Theoretical Physics III, 
Universit\"{a}tsstr. 1, D\"{u}sseldorf 40225, Germany}
\author{Anton Trushechkin}
\email{anton.trushechkin@hhu.de}
\affiliation{Heinrich Heine University D\"{u}sseldorf,
Faculty of Mathematics and Natural Sciences, Institute for Theoretical Physics III, 
Universit\"{a}tsstr. 1, D\"{u}sseldorf 40225, Germany}
\affiliation{Technical University of Braunschweig,  Faculty of Electrical Engineering, Information Technology,
and Physics, Institute for Communications Technology, Schleinitzstr. 22,  Braunschweig 38106, Germany}

\date{\today}

\begin{abstract}

We investigate the task of conference key agreement in near-term quantum networks, where the nodes are connected by sources of bipartite entangled states, under the class of local operations not requiring quantum memory. We derive upper bounds on the distillable conference key depending on the network topology and degree of entanglement of the sources, and prove tightness of these bounds for some particular cases. In these cases, we show that pairwise bipartite key distillation followed by merging the bipartite keys into the conference key is optimal.

\end{abstract}

\maketitle 

\section{Introduction}

Quantum networks and even the ``quantum internet'' are believed to be a part of future quantum communication infrastructure \cite{Kimble,QIntenet-Pirandola,QIntenet-Simon,QIntenet-Wehner,QInternet-book,QInternet2025,QPhotNet2025,QInternet2025Tech}. Secure distribution of cryptographic keys is supposed to be one of the earliest applications of quantum networks.

An important network cryptographic task is conference key agreement (CKA), i.e., establishing a common secret key for all participants in the network or a subset of them. Quantum conference key agreement (QCKA) has recently attracted attention  \cite{Shi2013MultiPartyQKA, Shukla2014BellQKA, Epping2017, Murta2020Review,  Grasselli2021AnonymousCKA}. One way of doing QCKA is distributing multipartite entangled states, e.g., GHZ states. It has been shown that directly utilizing multipartite entangled states (e.g. GHZ states) can improve the rate of conference key generation compared to pairwise schemes \cite{Epping2017,Murta2020Review,Pickston2023NetworkQCKA}.
However, the generation of multipartite GHZ states remains technologically challenging. 
Near-term quantum networks will likely be limited to pairwise entanglement, e.g. Bell states, between nodes. In such pair-entangled networks (PENs) \cite{Vicente2022}, the global resource state is composed solely of bipartite entangled states shared among various parties (nodes).

In this paper, we do not consider full QCKA protocols, which assume that the adversary can replace a source of entangled states (or a quantum channel) by their own, but consider the simplified task of \textit{conference key distillation} \cite{DW2005,ChristandlThesis,Augusiak2009,Das2021}, where the quantum state is given and known. The role of the adversary is that they hold the purification of this state (if it is mixed) and listen to the public classical communication.

We assume that users at their respective nodes can perform collective measurements on the subsystems coming from different sources. In the case of photons, this can be achieved using suitable interferometers. But we assume that the users cannot perform general local operations and classical  communications (LOCC), which would require quantum memories. In our scenario, users are not required to store their quantum systems after receiving the distributed quantum states. However, they are supposed to have preshared randomness, which gives rise to  the ``local operations and shared randomness'' (LOSR) scenario. The LOSR network setting, i.e., PEN states followed by LOSR operations, is actively studied \cite{Buscemi2004, Navascues2020, Kraft2021, Hansenne2022, Neumann2025_NoAdvantageLOSR}. 
Additionally, we allow postprocessing with arbitrary classical communication.

Although LOSR networks are severely restricted operationally, they can nevertheless generate nontrivial multipartite correlations. In particular, while the global resource consists only of bipartite states distributed along the edges of a network, the resulting network state may exhibit genuine multipartite entanglement \cite{Navascues2020,Neumann2025_NoAdvantageLOSR}. This observation suggests that there might exist multipartite strategies for QCKA which  outperform 
strategies restricted to biseparable source states \cite{Carrara2021, Wooltorton}.
In this work, we show that this is not the case for this type of network.

The central question we address are
upper bounds on the distillable conference key originating from the network topology and their tightness. Some general upper bounds are known \cite{Das2021,Pirandola2020}. We apply them to PEN networks and show that they reduce to the ``weakest cut''  bound. Roughly speaking, this means that, if the network is decomposed into two weakly connected parts, then the conference key agreement cannot exceed the capacity of this link. This bound is often not tight. We then derive bounds taking into account more general network topologies and bottleneck structures, which are based on arbitrary partitions of the corresponding graph, not just bipartitions.

We show that the protocol based on the spanning-tree packing \cite{PIN2010,PerfectOmni,TrushSpanTrees} allows us to achieve upper bounds for some particular cases, proving the tightness of these bounds in these cases. In contrast to  Ref.~\cite{TrushSpanTrees}, here we focus on upper bounds rather than on explicit protocols. Furthermore, in Ref.~\cite{TrushSpanTrees} only the scenario where the participants have already established a bipartite secret key is considered. Here we consider a more general situation of a network of bipartite entangled states. 

This work is organized as follows. In Section~\ref{SecNotations}, we give the definitions of the distillable conference key for an arbitrary multipartite resource state and the definition of the PEN states. In Section~\ref{SecDW}, we study conference key distillation protocols in which a single party serves as the  reference party.
This setting naturally gives rise to the tripartite generalization of the Devetak--Winter formula \cite{Devetak2005DWRate,augusiak2009multipartite,Epping2017}. We  derive a general  upper bound for the multipartite Devetak-Winter conference key generation rate in the LOSR setting and an upper bound for the tripartite generalization of the BB84 quantum key distribution (QKD) protocol \cite{Grasselli2018}.

It is known that genuine multipartite entanglement is not a necessary condition for QCKA \cite{Carrara2021, Wooltorton}, but upper bounds on the distillable conference key under restricted classes of multipartite states like LOSR preparable states from bipartite sources are unknown.

Starting from Section~\ref{SecCut}, we analyze the general conference key distillation. We discuss the bound based on vertex bipartitions (``weakest cut'') and also prove that the bound from Ref.~\cite{Das2021} applied to PEN states reduces to this bound. Section~\ref{SecPEN} contains general upper bounds. Our conclusions are given in Section~\ref{SecConcl}.

\section{Preliminaries}
\label{SecNotations}

\subsection{Distillable conference key}

Throughout the paper we fix the following conventions:\\
We work with $N$ honest parties $1,2,...,N$. We denote the corresponding quantum systems as $A_1, A_2, \dots, A_N$ and write $\vec A=(A_1,\ldots,A_N)$. 
At the start of every protocol round a source distributes a multipartite state $\rho_{A_1\ldots A_N}\equiv\rho_{\vec A}$, i.e., a density operator in a Hilbert space $\mathcal H_{A_1}\otimes\ldots\otimes\mathcal H_{A_N}\equiv \mathcal H_{\vec A}$.
 For every non-empty subset $I\subset\{1,\ldots,N\}\equiv[N]$ we abbreviate $A_I=(A_i)_{i\in I}$, $\overline I=[N]\backslash I$, $\mathcal H_I=\bigotimes_{i\in I}\mathcal H_{A_i}$, and  
    \begin{equation}    \rho_{A_I}=\Tr_{A_{\overline I}}\rho_{\vec A},\qquad \rho_{A_IE}=\Tr_{A_{\overline I}}\rho_{\vec AE}.
    \end{equation}
Here $\rho_{A_IE}$ is a purification of state $\rho_{A_I}$, i.e., a projector onto a vector in the Hilbert space $\mathcal H_{\vec A}\otimes\mathcal H_E$. The additional party $E$ is called Eve.
We consider the following scenario. First, the participants perform measurements from the LOSR (local operations and shared randomness) class, i.e., measurements  of the form
\begin{equation}
\label{EqPOVMlambda}
    \Big\{
    \sum_{\lambda\in\mathcal L}
    p_\lambda
    M_{x_1}^{(1),\,(\lambda)}
    \otimes
    M_{x_2}^{(2),\,(\lambda)}
    \otimes
    \ldots
    \otimes
    M_{x_N}^{(N),\,(\lambda)}
    \Big|\,
    x_i\in\mathcal X_i
    \Big\},
\end{equation}
where $\{M_{x_i}^{(i),(\lambda)}\}_{x_i\in\mathcal X_i}$, $i\in\{1,\ldots,N\}$, $\lambda\in\mathcal L$ are local positive operator-valued measures (POVM) depending on the parameter $\lambda$ randomly chosen from a finite set $\mathcal L$ according to the probability distribution $\{p_\lambda\}$  (public shared randomness). Here $\mathcal X_i$ are arbitrary finite sets of the outcomes. Then, the participants perform classical postprocessing. We denote this class of quantum channels (CPTP maps) as LOSR+PP.

The distillable conference key for the parties $I$ (following Ref.~\cite{PerfectOmni}, we call them ``secrecy-seeking parties'', while $[N]\backslash I$ are commonly referred to as ``helpers'') is defined as \cite{ChristandlThesis,Augusiak2009,Das2021} 
\begin{multline}
\label{EqKeyRate}
\overline r_{\rm key}(I)=\lim_{\varepsilon\to0}
\sup_{
\begin{smallmatrix}
n,m,\\
\Lambda\in{\rm LOSR+PP}
\end{smallmatrix}
}\Big\{\frac mn\Big|\:\frac12
\big\|\Lambda(\rho_{\vec AE}^{\otimes n})
\\-\rho^{(m), \,{\rm ideal}}_{K_IE^nC}\big\|_1
\leq\varepsilon\Big\},
\end{multline}
where 
$E^n$ is Eve's register originating from $n$ copies of $\rho_{\vec AE}$, the channel $\Lambda$ acts on the parties $\vec A$ (i.e., the action on $E^n$ is trivial), $C$ is the register for the classical communication in the postprocessing part of an LOSR+PP protocol (and also includes the public shared randomness $\lambda$), and
\begin{equation}
\label{EqKeyIdeal}
\begin{split}
\rho^{(m), \,{\rm ideal}}_{K_IE^nC}
&=2^{-m}
\sum_{k\in\{0,1\}^m}
(\ket{k}\bra{k})^{\otimes|I|}
\otimes\rho_{E^nC},\\
\rho_{E^nC}&=\Tr_{K_I}\Lambda(\rho_{\vec AE}^{\otimes n}).
\end{split}
\end{equation}
Here $K_i$, $i\in I$, are registers containing the secret key and associated with the corresponding parties. Each of them corresponds to the $2^m$-dimensional Hilbert space of key   and $\{\ket{k}\}$ is an orthonormal basis in this space. 
That is, in the ideal state (\ref{EqKeyIdeal}), the parties from $I$ share a perfectly classically correlated state uncoupled from Eve and classical communication. In the definition of \eqref{EqKeyRate} we assume that the helpers are trusted: We do not require privacy with respect to them. In classical multiterminal information theory, both problem statements with trusted and untrusted helpers are considered \cite{Csiszar2004}. We will return to this discussion in the end of the paper in Remark~\ref{RemHelpers}.

\subsection{Network model}

We consider a special class of multipartite states, namely the so called pair-entangled network (PEN) states~\cite{Vicente2022}.
Let a connected graph $([N],\mathcal E)$ be given, where $[N]=\{1,\ldots,N\}$ is the set of vertices and $\mathcal E$ is the set of edges. A PEN state has the form
\begin{equation}\label{EqPEN}
    \rho_{\vec A}=\bigotimes_{e\in\mathcal E}\rho_e,
\end{equation}
where $\rho_e$ are (possibly different) bipartite states acting on Hilbert spaces $\mathcal H_e$ corresponding to the edges $e$. An example is depicted in Fig.~\ref{FigPEN}.

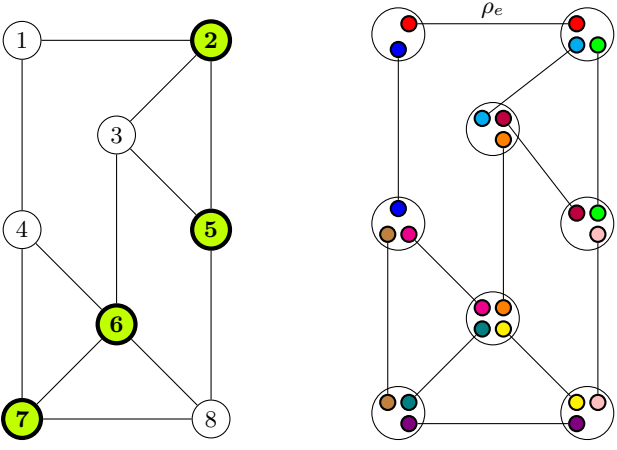
\begin{figure}
    \begin{tikzpicture}[node distance=2.5cm]
        
\tikzstyle{vertex}=[circle,draw,inner sep=0pt,minimum size=5mm]
\tikzstyle{mainvertex}=[circle,draw,inner sep=0pt,minimum size=5mm,ultra thick,font=\bfseries,fill=lime]

\node [vertex] (a1) {1};
\node [mainvertex] (a2) [right of=a1] {2};
\node [vertex] (a4) [below of=a1] {4};
\node [mainvertex] (a7) [below of=a4] {7};
\node [vertex] (a8) [right of=a7] {8};
\node [mainvertex] (a5) [below of=a2] {5};
\node [vertex] (a3) at (a1) [xshift=1.25cm, yshift=-1.25cm] {3};
\node [mainvertex] (a6) [below of=a3] {6};

\draw (a1) -- (a2);
\draw (a1) to (a4);        
\draw (a2) to (a3);        
\draw (a2) to (a5);        
\draw (a3) to (a6);        
\draw (a3) -- (a5);
\draw (a4) -- (a6);
\draw (a4) -- (a7);
\draw (a5) to (a8);        
\draw (a6) -- (a7);
\draw (a6) -- (a8);
\draw (a7) -- (a8);


\end{tikzpicture}
    \qquad\qquad\quad
    \begin{tikzpicture}[node distance=2.5cm]
        
\tikzstyle{vertex}=[circle,draw,inner sep=0pt,minimum size=7mm]
\tikzstyle{subvertex}=[circle,draw,inner sep=0pt,minimum size=2mm,thick]

\node [vertex] (a1) {};
\node [vertex] (a2) [right of=a1] {};
\node [vertex] (a4) [below of=a1] {};
\node [vertex] (a7) [below of=a4] {};
\node [vertex] (a8) [right of=a7] {};
\node [vertex] (a5) [below of=a2] {};
\node [vertex] (a3) at (a1) [xshift=1.25cm, yshift=-1.25cm] {};
\node [vertex] (a6) [below of=a3] {};

\node (a11) at (a1) [subvertex,fill=red,xshift=0.14cm,yshift=0.14cm] {};
\node  (a12) at (a1) [subvertex,fill=blue,yshift=-0.2cm] {};

\node (a21) at (a2) [subvertex,fill=red,xshift=-0.14cm,yshift=0.14cm] {};
\node (a22) at (a2)
[subvertex,fill=cyan,xshift=-0.14cm,yshift=-0.14cm] {};
\node (a23) at (a2)
[subvertex,fill=green,xshift=0.14cm,yshift=-0.14cm] {};

\node (a31) at (a3) [subvertex,fill=cyan,xshift=-0.14cm,yshift=0.14cm] {};
\node (a32) at (a3) [subvertex,fill=purple,xshift=0.14cm,yshift=0.14cm] {};
\node (a33) at (a3) [subvertex,fill=orange,xshift=0.14cm,yshift=-0.14cm] {};

\node (a41) at (a4)
[subvertex,fill=blue,yshift=0.2cm] {};
\node (a42) at (a4)
[subvertex,fill=magenta,xshift=0.14cm,yshift=-0.14cm] {};
\node (a43) at (a4)
[subvertex,fill=brown,xshift=-0.14cm,yshift=-0.14cm] {};

\node (a51) at (a5)
[subvertex,fill=green,xshift=0.14cm,yshift=0.14cm] {};
\node (a52) at (a5)
[subvertex,fill=purple,xshift=-0.14cm,yshift=0.14cm] {};
\node (a53) at (a5)
[subvertex,fill=pink,xshift=0.14cm,yshift=-0.14cm] {};

\node (a61) at (a6)
[subvertex,fill=magenta,xshift=-0.14cm,yshift=0.14cm] {};
\node (a62) at (a6)
[subvertex,fill=orange,xshift=0.14cm,yshift=0.14cm] {};
\node (a63) at (a6)
[subvertex,fill=yellow,xshift=0.14cm,yshift=-0.14cm] {};
\node (a64) at (a6)
[subvertex,fill=teal,xshift=-0.14cm,yshift=-0.14cm] {};

\node (a71) at (a7)
[subvertex,fill=brown,xshift=-0.14cm,yshift=0.14cm] {};
\node (a72) at (a7)
[subvertex,fill=teal,xshift=0.14cm,yshift=0.14cm] {};
\node (a73) at (a7)
[subvertex,fill=violet,xshift=0.14cm,yshift=-0.14cm] {};

\node (a81) at (a8)
[subvertex,fill=yellow,xshift=-0.14cm,yshift=0.14cm] {};
\node (a82) at (a8)
[subvertex,fill=pink,xshift=0.14cm,yshift=0.14cm] {};
\node (a83) at (a8)
[subvertex,fill=violet,xshift=-0.14cm,yshift=-0.14cm] {};

\draw (a11) to node [above, inner sep=.8mm] {$\rho_e$} (a21);
\draw (a12) -- (a41);
\draw (a22) -- (a31);
\draw (a23) -- (a51);
\draw (a33) -- (a62);
\draw (a32) -- (a52);
\draw (a42) -- (a61);
\draw (a43) -- (a71);
\draw (a53) -- (a82);
\draw (a64) -- (a72);
\draw (a63) -- (a81);
\draw (a73) -- (a83);

\end{tikzpicture}
    \caption{Left: Example for the graph of a network. The vertices (or nodes) correspond to participants and the edges correspond to bipartite entangled states between them. Vertices from the subset $I$ of participants who want to establish a common conference key (``secrecy-seeking parties'') are depicted in green. The other participants can assist them: Here the vertices 5 and 6 are connected only via the vertices outside $I$. Right: more detailed structure of the network and a PEN state. The ``total'' Hilbert space is the tensor product of the spaces corresponding to the participants. The Hilbert spaces of the participants are the tensor products of their corresponding parts of the bipartite states.
    }
    \label{FigPEN}
\end{figure}

A pure PEN state has the form 
\begin{equation}
\label{EqPENpure}
    \rho_{\vec A}=
    \bigotimes_{e\in\mathcal E}
    \ket{\psi_e}\bra{\psi_e}
\end{equation} 
If $\rho_{\vec A}$ is a mixed state, then, for each $\rho_e$ in Eq.~(\ref{EqPEN}), there exists a purification $\ket{\psi_{e,E_e}}
\bra{\psi_{e,E_e}}$, with $E_e$ being the purifying system.
The purification $\rho_{\vec AE}$ of the state $\rho_{\vec A}$ is given by
\begin{equation}
    \rho_{\vec AE}=\bigotimes_{e\in\mathcal E}
    \ket{\psi_{e,E_e}}
    \bra{\psi_{e,E_e}}.
\end{equation}

\section{Bounds on the multipartite Devetak-Winter key rate}
\label{SecDW}
Let us first explore a scenario when there is a distinguished party $A_1$ whose raw key is considered as the reference  and the other participants try to fix discrepancies between their raw keys and that of $A_1$ (treated in this scenario as ``errors''). A known achievable conference key rate for this case is the multipartite generalization of the Devetak-Winter formula \cite{Devetak2005DWRate,augusiak2009multipartite, Epping2017}:
\begin{align}
\label{EqDW}
    r_{\rm DW}=\max_{\begin{smallmatrix}
        \text{LOSR}\\
        \text{measurements}
    \end{smallmatrix}}
    \lbrace
    \min_{i=2, \dots, N}I(X_1:X_i)-I(X_1:E)
    \rbrace
\end{align}
where $I(X_1:X_i)$ is the classical mutual information between the measurement outcomes $X_1$ and $X_i$ of the parties $A_1$ and $A_i$, respectively, and $I(X_1:E)$ is the quantum mutual information between  $X_1$ and the eavesdropper Eve ($E$) \cite{augusiak2009multipartite}.
Then the following bound holds for the multipartite Devetak-Winter  key rate for  PEN states.

 \begin{theorem}
 \label{DWbound}
   Let $\rho$ be a PEN state. Then the multipartite Devetak-Winter key rate (\ref{EqDW}) for $N$ participants is upper bounded by 
    \begin{align}
    \label{Eq:devetakwinter}
        r_{\rm DW}([N]) \le \frac{S(A_1)}{N-1},
    \end{align}
    where $S(A_1)$ is the von Neumann entropy of the reduced state of $A_1$. In  the case of $\rho$ being a pure state (see Eq. \eqref{EqPENpure}) and $A_1$ connected to all other users, this bound can be reached by a bipartite QKD strategy.
\end{theorem}
The proof can be found in Appendix \ref{proofDWbound}.
The theorem shows that no conference key agreement protocol should outperform the simple bipartite QKD protocol, where one party constructs a private key with each other party individually and merges them into a final mutual key shared by all parties.

As an example, consider the PEN state corresponding to the graph depicted in the left part of Fig.~\ref{FigTh}, where each edge corresponds to one Bell pair. Then $S(A_1)=2$ and Theorem~\ref{DWbound} gives $r_{\rm DW}\leq 1$.
To see that the bound is tight let $A_1$ run a perfect bipartite QKD strategy separately with $A_2$ and $A_3$, producing one-bit keys $K_{12}$ and $K_{13}$. If we take the conference key to be $K:=K_{12}$ and let $A_1$ publicly announce $K_{12}\oplus K_{13}$, $A_3$ can reconstruct $K_{12}$, so all three parties share the same secret bit at rate $1$, and the announcement reveals nothing about $K$ to Eve.

Let us impose a further restriction on the scenario and consider the multipartite version of the BB84 protocol \cite{Epping2017, Grasselli2021AnonymousCKA}. Ideally, this protocol assumes a GHZ state (not a PEN state) distributed at each round. However, using PEN states and LOSR maps, states that can be interpreted as noisy versions of the GHZ state can be generated. Upper bounds on fidelities of such states with the GHZ state can be found in Ref.~\cite{Navascues2020_GNME,Neumann2025_NoAdvantageLOSR}. First, the users map their systems into qubits using LOSR transformations and then measure them and postprocess measurements according to the multipartite BB84 protocol. We denote the corresponding rate as $r_{\rm BB84}$.

Interestingly, in Ref.~\cite{Carrara2021,Wooltorton}, it was shown that even biseparable qubit states can lead to nonzero conference key rate, i.e., genuine multipartite entanglement is not a prerequisite for conference key agreement. 
On the other hand, LOSR transformations of PEN states lead to a strictly larger class of quantum states including genuinely multipartite entangled states with larger GHZ state fidelity.

In the following theorem, however, we show that $r_{\rm BB84}$ cannot outperform the conference key rate for biseparable states. 
\begin{theorem}
\label{BB84bound}
Let $\rho$ be a tripartite qubit state which can be prepared in PEN with 3 nodes (see the left part of Fig.~\ref{FigTh}). Then the maximal rate  for  conference key agreement with a multipartite BB84 protocol is
    \begin{align}
    \label{EqBB84rate}
    \max_{\rho\in{\rm PEN}}
    r_{\rm BB84}([N=3])=1-h\left(\frac{1}{4}\right) \approx 0.188,
\end{align}
where $h$ is the binary entropy function.
\end{theorem}

The proof can be found in Appendix \ref{proofBB84bound}. The right-hand side of Eq.~(\ref{EqBB84rate}) can be achieved by biseparable states, which was shown in Ref.~\cite{Carrara2021}.

\section{Bounds on the distillable conference key  based on bipartitions}
\label{SecCut}
In the previous section we considered the multipartite Devetak-Winter key rate formula for many participants. We mentioned that this formula corresponds to the scenario with one distinguished participant. If the network is decentralized, this is too restrictive. 
More general key agreement protocols of conference key agreement by public discussion in decentralized networks were studied in classical information theory \cite{Csiszar2004,PIN2010,PerfectOmni}. Consider the PEN state corresponding to the graph depicted in the left part of Fig.~\ref{FigPEN}, where each edge corresponds to the generation of one Bell pair per round. We mentioned that the upper bound for the multipartite Devetak-Winter key rate is 1. In the following we will show that the actual distillable conference key here is 3/2.

We start with an obvious bound on the distillable conference key for pure PEN states (\ref{EqPENpure}). Recall that a cut of a graph is a set of edges whose removal breaks the graph into two disconnected parts (sides of the cut) \cite{Diestel}. We say that a  cut is $I$-proper if each of its sides contains at least one vertex from the set $I\subset [N]$.

\begin{proposition}
\label{PropCut}
    Consider a pure PEN state (\ref{EqPENpure}) and denote $S_e$ the entanglement entropy of the edge bipartite states $\ket{\psi_e}$. Then
    \begin{equation}
    \label{EqCutBnd}
        \overline r_{\rm key}(A_I)\leq\min_{I\text{-proper cut } C}
        \sum_{e\in C}S_e,
    \end{equation}
    where the minimum is taken over all $I$-proper cuts $C$. 
\end{proposition}

Thus, the distillable conference key is upper bounded by the total entanglement entropy of the ``weakest'' cut.

\begin{proof}

    Consider the $I$-proper cut  $C$. Consider the problem of bipartite secret key distillation between the two subsets in this partition, i.e.,  we treat all vertices from each set as one ``aggregated'' participant and we allow them to perform collective operations. If a conference key for the subset $I$ is established, then, obviously, it can be viewed as a bipartite secret key between these two subsets. Hence, the distillable conference key  cannot exceed this distillable bipartite key. As is well-known, the bipartite distillable key is upper bounded by the relative entropy of entanglement \cite{ChristandlThesis,WildeBook}, which, in turn, in the case of a pure state, is equal to the entanglement entropy. For a pure PEN-state (\ref{EqPENpure}) considered as a bipartite state between the two subsets in a bipartition of $[N]$, the entanglement entropy is the sum of the bipartite entanglement entropies $S_e$ over the edges $e$ from the cut  (Fig.~\ref{FigCut}). Minimization over the $I$-proper cuts gives the right-hand side of Eq.~(\ref{EqCutBnd}).
\end{proof}

\begin{figure}
    \begin{tikzpicture}[node distance=2cm]
        
\tikzstyle{vertex}=[circle,draw,inner sep=0pt,minimum size=5mm]
\tikzstyle{mainvertex}=[circle,draw,inner sep=0pt,minimum size=5mm,ultra thick,font=\bfseries,fill=lime]

\node [mainvertex] (a3) {3};
\node [mainvertex] (a2) [right of=a3] {2};
\node [mainvertex] (a1) [above of=a3, xshift=1cm, yshift=-.5cm] {1};

\draw (a1) -- (a2);
\draw (a2) to (a3);        
\draw (a1) to (a3);        

\draw [thick, red] (0,.75) -- (2,.75);

\end{tikzpicture}
    \qquad\qquad\quad
    \begin{tikzpicture}[node distance=2.5cm]
        
\tikzstyle{vertex}=[circle,draw,inner sep=0pt,minimum size=5mm]
\tikzstyle{mainvertex}=[circle,draw,inner sep=0pt,minimum size=5mm,ultra thick,font=\bfseries,fill=lime]

\node [vertex] (a1) {1};
\node [mainvertex] (a2) [right of=a1] {2};
\node [vertex] (a4) [below of=a1] {4};
\node [mainvertex] (a7) [below of=a4] {7};
\node [vertex] (a8) [right of=a7] {8};
\node [mainvertex] (a5) [below of=a2] {5};
\node [vertex] (a3) at (a1) [xshift=1.25cm, yshift=-1.25cm] {3};
\node [mainvertex] (a6) [below of=a3] {6};

\draw (a1) -- (a2);
\draw (a1) to (a4);        
\draw (a2) to (a3);        
\draw (a2) to (a5);        
\draw (a3) to (a6);        
\draw (a3) -- (a5);
\draw (a4) -- (a6);
\draw (a4) -- (a7);
\draw (a5) to (a8);        
\draw (a6) -- (a7);
\draw (a6) -- (a8);
\draw (a7) -- (a8);

\draw [red, thick] (-.25,-1.25) -- (2.75,-3.75);

\end{tikzpicture}
    \caption{Illustration of Proposition~\ref{PropCut}. Here each edge corresponds to one Bell pair. A simple bound on the distillable conference  rate is given by a ``weakest'' cut: If we find a cut of the graph such that both parts contain vertices from $I$ (secrecy-seeking parties, depicted in green), then the distillable conference key cannot be larger than the distillable bipartite key corresponding to this bipartition, which is equal to the total entanglement entropy of the edges from the cut. The weakest cuts depicted in the figure give the bound $\overline r\leq2$ for the network to the left and the bound $\overline r\leq3$ for the network to the right. However, these bounds are not tight, see Fig.~\ref{FigTh}.
    }
    \label{FigCut}
\end{figure}
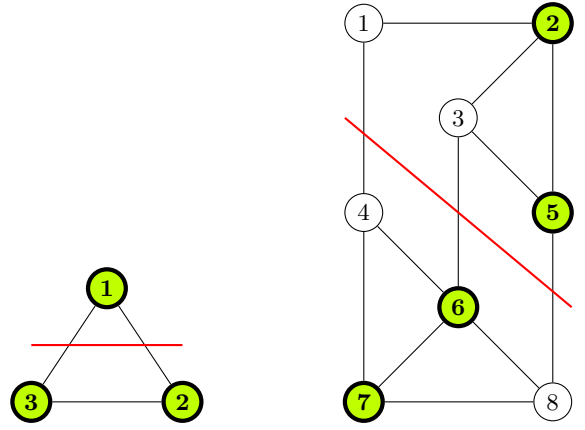

In Ref.~\cite{Das2021}, the following upper bound for the distillable conference key for the case $I=[N]$ (i.e., all participants want to have a conference key) was proved:

\begin{equation}
\label{EqDas}
    r_{\rm key}(\vec A)
\leq\lim_{n\to\infty}\frac1n
\min_{\sigma\in{\rm BISEP}}
D(\rho^{\otimes n}_{\vec A}\|\sigma).
\end{equation}
 The quantity $\min_{\sigma\in{\rm BISEP}}D(\rho\|\sigma)$ is called the relative entropy of genuine multipartite entanglement (GME) of $\rho$ and the right-hand side of Ineq.~(\ref{EqDas}) is its regularized version.

In the case of pure PEN states, the bound (\ref{EqDas}) is reduced to the bound (\ref{EqCutBnd}). Namely, the following statement, which we prove in Appendix~\ref{SecDasCutProof}, holds.

\begin{theorem}
\label{PropDasCut}
    Given a pure PEN state $\rho_{\vec A}$, 
    \begin{equation}
        \min_{\sigma\in{\rm BISEP}}D(\rho\|\sigma)
        =
        \min_{I\text{-proper cut } C}
        \sum_{e\in C}S_e.
    \end{equation}
\end{theorem}

Thus, for PEN states, the bound (\ref{EqDas}) reduces to the weakest cut (vertex bipartition). 
Fig.~\ref{FigCut} illustrates the weakest-cut bound (\ref{EqCutBnd}), but, as we will see in the next sections and in Fig.~\ref{FigTh}, these bounds are not tight. Namely, for the triangle graph the ``weakest'' cut gives an upper bound of 2 for the distillable conference key, while the actual distillable conference key is (as we will see) 3/2.

Let us also note that upper bounds on the distillable conference key based on multipartite generalizations of squashed entanglement were proposed \cite{MultiSqEntHoro,MultiSqEntWilde}. However, these  bounds applied to the PEN states correspond to the ``weakest'' cut bound.

\section{Bounds on the distillable conference key based on arbitrary partition}
\label{SecPEN}

Consider a partition
\begin{equation}
\label{EqP}
    \mathcal P=\{J_1,\ldots,J_p\}
\end{equation}
of the set $[N]$ into disjoint subsets such that all $J_\alpha\cap I$, $\alpha=1,\ldots,p$, are nonempty.  We will refer to such partitions as $I$-\textit{proper} partitions (thus generalizing the $I$-proper bipartition from Sec.~\ref{SecCut}). Also denote $|\mathcal P|=p$ the number of the partition sets.

Let us illustrate the condition on a proper partition. Let $N=4$ and $I={1,2,3}$. Then the partition of the set $\{1,2,3,4\}$ into the subsets $\{1,2\}$ and $\{3,4\}$ is $I$-proper, while the partition into the subsets $\{1,2\}$, $\{3\}$, and $\{4\}$ is not proper because the last subset does not intersect with $I$. Examples of proper partitions are depicted in Fig.~\ref{FigTh}.

\begin{figure}
    \begin{tikzpicture}[node distance=2cm]
        
\tikzstyle{vertex}=[circle,draw,inner sep=0pt,minimum size=5mm]
\tikzstyle{mainvertex}=[circle,draw,inner sep=0pt,minimum size=5mm,ultra thick,font=\bfseries,fill=lime]

\node [mainvertex] (a3) {3};
\node [mainvertex] (a2) [right of=a3] {2};
\node [mainvertex] (a1) [above of=a3, xshift=1cm, yshift=-.5cm] {1};

\draw [ultra thick,red] (a1) -- (a2);
\draw [ultra thick,red] (a2) to (a3);        
\draw [ultra thick,red] (a1) to (a3);        


\draw[blue,thick,dashed] (0,0) circle  (0.45cm);

\draw[blue,thick,dashed] (2,0) circle  (0.45cm);

\draw[blue,thick,dashed] (1,1.5) circle  (0.45cm);

\end{tikzpicture}
    \qquad\qquad
    \begin{tikzpicture}[node distance=2.5cm]
        
\tikzstyle{vertex}=[circle,draw,inner sep=0pt,minimum size=5mm]
\tikzstyle{mainvertex}=[circle,draw,inner sep=0pt,minimum size=5mm,ultra thick,font=\bfseries,fill=lime]

\node [vertex] (a1) {1};
\node [mainvertex] (a2) [right of=a1] {2};
\node [vertex] (a4) [below of=a1] {4};
\node [mainvertex] (a7) [below of=a4] {7};
\node [vertex] (a8) [right of=a7] {8};
\node [mainvertex] (a5) [below of=a2] {5};
\node [vertex] (a3) at (a1) [xshift=1.25cm, yshift=-1.25cm] {3};
\node [mainvertex] (a6) [below of=a3] {6};

\draw (a1) -- (a2);
\draw [ultra thick,red] (a1) to (a4);        
\draw [ultra thick,red] (a2) to (a3);        
\draw [ultra thick,red] (a2) to (a5);        
\draw [ultra thick,red] (a3) to (a6);        
\draw (a3) -- (a5);
\draw (a4) -- (a6);
\draw (a4) -- (a7);
\draw [ultra thick,red] (a5) to (a8);        
\draw (a6) -- (a7);
\draw (a6) -- (a8);
\draw (a7) -- (a8);

\draw[blue,thick,dashed] (1.25,0) ellipse  (1.75cm and .47cm);

\draw[blue,thick,dashed,rotate=-45] (2.65,-.01) ellipse  (1.3cm and .5cm);

\draw[thick,rounded corners=.4cm,dashed,blue]
(-.4,-1.6) -- (-.4,-5.4) -- (3.4,-5.4) -- cycle;



\end{tikzpicture}
    \caption{Illustration of Theorems~\ref{ThPENpure} and \ref{ThPEN}. These theorems include the minimization over all vertex partitions such that each subset intersects with $I$ (secrecy-seeking parties, depicted in green). We consider arbitrary vertex partitions of the graphs and the partition subsets are encircled by blue dashed curves. For PEN states, the total entanglement entropy (if the state is pure) or the total entanglement cost (if the state is mixed) of the cross edges (i.e., edges connecting different partition subsets, depicted in red) gives an upper bound for the distillable conference key (see Eq. \eqref{EqPENpurebnd}). If each edge corresponds to a Bell pair, then, according to Theorem~\ref{ThPENpure}, the partitions depicted here give stronger bounds than the estimates based on the weakest cuts (see Fig.~\ref{FigCut}):  $\overline r\leq3/2$ for the network to the left and the bound $\overline r\leq5/2$ for the network to the right.
    }
    \label{FigTh}
\end{figure}
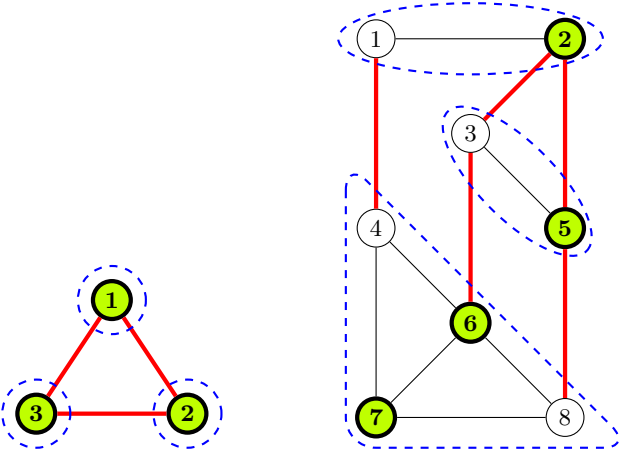

For a vertex partition $\mathcal P$, denote $\mathcal E(\mathcal P)$ the set of edges connecting vertices from different subsets in the vertex partition $\mathcal P$.
\begin{theorem}
\label{ThPENpure}
    Let 
    \begin{equation}
        \rho_{\vec A}=\bigotimes_{e\in\mathcal E}\ket{\psi_e}\bra{\psi_e}
    \end{equation}
    be a pure PEN state [cf. Eq.~(\ref{EqPEN})]. Then
    \begin{equation}
    \label{EqPENpurebnd}
        \overline r_{\rm key}(A_I)\leq
        \min_{\mathcal P}
        \frac{1}{|\mathcal P|-1}
        \sum_{e\in\mathcal E(\mathcal P)}S_e,
    \end{equation}
    where $S_e$ is the entanglement entropy of the bipartite state $\ket{\psi_e}$ and the minimum is taken over $I$-proper partitions $\mathcal P$.
\end{theorem}

The upper bound (\ref{EqPENpurebnd}) is depicted in Fig.~\ref{FigTh}. The proof is based on the estimation of the maximal total correlation (a multipartite generalization of the mutual information) of the outcomes of local measurements (see Appendix \ref{proofThPENpure}).

\begin{remark}
    If $|I|=N$, i.e., all participants are secrecy-seeking, then the bound \eqref{EqPENpurebnd} is tight and achieved by the optimal spanning-tree-packing protocol, see \cite{TrushSpanTrees} and  papers from classical information theory \cite{PIN2010,PerfectOmni}. Namely, for each edge $e$, the distillable bipartite secret key for the state $\ket{\psi_e}$ is equal to $S_e$. Then, to get a conference key from the resulting network of bipartite secret keys, we use the optimal secret key packing. This means that, in this case, collective measurements of states coming from different sources or collective postprocessing do not give advantage: Doing pairwise bipartite QKD and then merging the  bipartite keys into a conference key is an optimal strategy in PEN networks.
\end{remark}

We can generalize Theorem~\ref{ThPENpure} to the case of arbitrary (mixed) PEN states as follows:

\begin{theorem}
    \label{ThPEN}
    Let $\rho_{\vec A}$
    be a PEN state (\ref{EqPEN}). Then
    \begin{equation}
    \label{EqKeyRateEF}
        \overline r_{\rm key}(A_I)\leq
        \min_{\mathcal P}
        \frac{1}{|\mathcal P|-1}
        \sum_{e\in\mathcal E(\mathcal P)}E_{\rm F}^\infty(\rho_e),
    \end{equation}
    where $E_{\rm F}^\infty$ denotes the regularized entanglement of formation and the minimum is taken over $I$-proper partitions $\mathcal P$.
\end{theorem}
The proof can be found in Appendix~\ref{proofThPEN}.
Recall that the entanglement of formation $E_{\rm F}(\rho_{AB})$ of the bipartite state $\rho_{AB}$ is defined as

\begin{equation}
    E_{\rm F}(\rho_{AB})
    =\inf\big\{\sum_\alpha p_\alpha S_{\rm ent}(\ket{\psi_\alpha}_{AB})\big\},
\end{equation}
where the infimum is taken over decompositions
\begin{equation}
\label{EqRhoABdecompose}
    \rho_{AB}=\sum_\alpha p_\alpha\ket{\psi_\alpha}_{AB}\!\bra{\psi_\alpha}
\end{equation}
and $S_{\rm ent}(\ket{\psi_\alpha}_{AB})$ denotes the entanglement entropy of the pure state $\ket{\psi_\alpha}_{AB}$. The regularized entropy of formation is defined as
\begin{equation}
    E_{\rm F}^\infty(\rho_{AB})
    =\lim_{n\to\infty}\frac1n
    E_{\rm F}(\rho_{AB}^{\otimes n}).
\end{equation}

Note that, in the bipartite scenario, it is known that the entanglement cost (equal to the regularized entanglement of formation) is an upper bound for the distillable key \cite{KoashiWinter2004,ChristandlThesis}. Appendix~\ref{proofThPEN} contains another  proof of this fact, which is based on typical sequences (sketched in Ref.~\cite{KoashiWinter2004}). 

Finally, let us make the following simple observation
about the distillable conference key for the case when the network graph is a tree, i.e., contains no cycles. 

\begin{proposition}
\label{PropTree}
    Let the network graph $(\mathcal V,\mathcal E)$ be a tree and $r_e$, $e\in\mathcal E$, be the bipartite distillable key for the (general) state $\rho_e$. Then, the distillable conference key
    $\overline r_{\rm key}(A_I)$ is given by 
    \begin{equation}
        \overline r(A_I)=
        \min_{e\in\widetilde{\mathcal E}}
        r_e,
    \end{equation}
    where $\widetilde{\mathcal  E}$ is the set of edges in the minimal subtree $(\widetilde{\mathcal  V},\widetilde{\mathcal  E})$ of the tree graph $([N],\mathcal E)$ for which $I\subset\widetilde{\mathcal V}$ (see Fig.~\ref{FigTree}). 
\end{proposition}

\begin{proof}
    From one side, we can repeat the proof of Proposition~\ref{PropCut} that the distillable conference key cannot increase the bipartite distillable key for any vertex bipartition. Removal of any edge breaks a tree into two disconnected parts. $I$-proper bipartitions correspond to removal of the edges from $\widetilde{\mathcal E}$. Hence, $\overline r(A_I)\leq r_e$ for all $e$ defining an $I$-proper bipartition. 
    
    From the other side, if $r_{e^*}$ is the minimal bipartite distillable key among the edges from $\widetilde{\mathcal E}$, then all other bipartite links from $\widetilde{\mathcal E}$ can work at least with the same distillation rate providing the conference distillation rate $r_{e^*}$.
\end{proof}

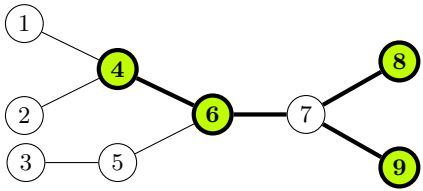
\begin{figure}
       \begin{tikzpicture}[node distance=7mm, bend angle=20]

\tikzstyle{vertex}=[circle,draw,inner sep=0pt,minimum size=5mm]

\tikzstyle{mainvertex}=[circle,draw,inner sep=0pt,minimum size=5mm,ultra thick,font=\bfseries,fill=lime]

        \node [mainvertex] (a6)  {6};
        \node [mainvertex] (a4) [left=of a6, yshift=6mm] {4};
        \node [vertex] (a5) [below=of a4] {5};
        \node [vertex] (a1) [left=of a4, yshift=6mm] {1};
        \node [vertex] (a2) [below=of a1] {2};
        \node [vertex] (a3) [left=of a5] {3};
        \node [vertex] (a7) [right=of a6] {7};
        \node [mainvertex] (a8) [right=of a7, yshift=7mm] {8};
        \node [mainvertex] (a9) [right=of a7, yshift=-7mm] {9};
        \draw (a1) -- (a4);
        \draw (a2) -- (a4);
        \draw (a3) -- (a5);
        \draw [ultra thick] (a4) -- (a6);
        \draw (a5) -- (a6);
        \draw [ultra thick] (a6) -- (a7);
        \draw [ultra thick] (a7) -- (a8);
        \draw [ultra thick] (a7) -- (a9);
    \end{tikzpicture}
    \caption{Tree network. The subgraph induced by the vertices 4 and 6--9 constitutes the minimal tree containing the picked vertices. Then the distillable conference key is equal to the minimal bipartite distillable key in this ``shortened tree''. Vertices 1, 2, 3, and 5 do not connect any picked vertices, hence, they cannot assist the distillation task and their edges are unimportant.}
    \label{FigTree}
\end{figure}

\begin{remark}
\label{RemHelpers}
Let us recall that we assumed the helpers to be trusted, i.e., they may also know the conference key or a part of it. Let us now consider the case of untrusted helpers who are not allowed to get any information about the conference key. Of course, the derived upper bounds are still valid in the case of increased restrictions. However, even in the networks like in Fig.~\ref{FigTree}, where the subsets $\{4,6\}$ and $\{8,9\}$ of the secrecy-seeking parties are separated by the helper node 7, generation of a conference key which is private also to helpers can be possible. In the classical case this is obviously impossible because any communication between these two subsets flows through the untrusted node 7. In the quantum case, the parties can distill a certain number of Bell pairs from all necessary bipartite links (under the condition that it is possible) and then the node 7 performs GHZ measurements creating GHZ correlations between the nodes 6, 8, and 9. This allows the secrecy-seeking parties to establish a conference key private also from the node 7. The node 6, 8 and 9 can verify the GHZ-type correlation such that node 7 cannot cheat. 

Thus, in the presence of untrusted nodes, genuinely multipartite quantum conference key agreement protocols (rather than putting together the bipartite ones) might be necessary. This aligns with the example of the advantage from the use of genuinely multipartite protocols from Ref.~\cite{Epping2017}. A more detailed analysis of conference key agreement in pair-entangled networks with untrusted nodes is a subject for future research. 

\end{remark}

\section{Conclusions and discussion}

\label{SecConcl}

We have studied the task of conference key agreement in pair-entangled networks (PEN), where the global network state consists of bipartite entangled states shared along the edges of a network graph. Networks we considered are memory-free: Each node measures its incoming quantum systems immediately after possible local operations with shared randomness (LOSR) and uses public classical communication afterwards.

We derived  an upper bound on the multipartite Devetak-Winter key rate implying that a simple pairwise QKD strategy saturates the bound on fully connected graphs. This aligns with broader evidence that highly resourceful multipartite states are hard to create in LOSR networks (see Theorem \ref{DWbound}). 

Next, we showed that for the multipartite BB84 protocol in a 3-node PEN, no LOSR state generated in the network outperforms the best biseparable resource (see Theorem \ref{BB84bound}). This already signals that straightforward GHZ-based approaches yield no rate advantage in PENs.

We then considered the most general conference key distillation protocols not requiring a central distinguished node and have derived upper bounds depending on the network topology and degree of entanglement of the source. Theorem~\ref{ThPENpure} gives us an upper bound for pure PEN states and Theorem~\ref{ThPEN} generalizes this bound to the case of mixed PEN states. These bounds are based on simple partition structures of the network and they have been shown to be stronger than previously known cut-based bounds.

Moreover, in the case of pure states and all nodes being the secrecy-seeking parties (no helpers), we have shown that the obtained bounds are tight and can be achieved by individual measurements followed  by the spanning-tree-packing algorithm of conference key propagation from previous work \cite{TrushSpanTrees} and also known in classical network information theory \cite{PIN2010,PerfectOmni}. In particular, this means that collective measurements and collective postprocessing do not give an advantage in this case and, although, in principle, many different multipartite strategies might be considered in a quantum network, the optimal approach is remarkably simple.

We have also provided a calculation of the relative entropy of genuine multipartite entanglement and showed that for pure pairwise entangled network states it is reduced to the minimal entanglement entropy of vertex bipartitions (Theorem~\ref{PropDasCut}). This result can be viewed as a multipartite generalization of the well-known fact that the relative entropy of (bipartite) entanglement of a pure bipartite state is equal to its entanglement entropy.

Two open questions can be suggested for the future. First, in the case of a pure PEN state, but in the presence of helpers, is the strategy of bipartite distillation protocols followed by merging the bipartite keys into a conference one still optimal? 
The second open question is whether the entropy of formation in Theorem~\ref{ThPEN} can be replaced by other upper bounds for bipartite distillable key, e.g., the (bipartite) relative entropy of entanglement or squashed entanglement.

\section*{Acknowledgments}

We thank Giacomo Carrara and Tulja Varun Kondra for helpful discussions. This work was funded by the Federal Ministry of Research, Technology and Space BMFTR (Project QuKuK, Grant No.~16KIS1618K). J.N., H.K., and D.B. also acknowledge support by Deutsche Forschungsgemeinschaft
(DFG, German Research Foundation) under Germany’s
Excellence Strategy -- Cluster of Excellence Matter and
Light for Quantum Computing (ML4Q) EXC 2004/1 --
390534769. A.T. acknowledges support by BMFTR (Project QR.N, Grant No.~16KIS2196).


\appendix

\fancyhf{}
\fancyhead[R]{\thepage}
\fancyhead[L]{\leftmark}

\section{Proof of Theorem~\ref{DWbound} }
\label{proofDWbound}
\begin{proof}[Proof of Theorem~\ref{DWbound}]
Consider a network with bipartite source states and local operations. Then
\begin{align}
    r_{DW} &\le \min_{i=2, \dots, N}\{I(A_1:A_i)\} \\
    &\le \frac{\sum_{i=1}^{N-1} I(A_1:A_i)}{N-1}\\
    & \le \frac{H(A)}{N-1},
\end{align}
where in the last step we use the result from \cite{chaves2015}. The tightness can be shown by performing a bipartite QKD strategy, where the party $A_1$ generates a perfect private key with the other parties separately and combines them into one single shared conference key. 
\end{proof}

\section{Proof of Theorem~\ref{BB84bound}}
\label{proofBB84bound}
The asymptotic key-rate in the BB84-protocol can be written as
\begin{align*}
     r_{\text{BB84}}=1-&h\left(\frac{1-\braket{X_AX_BX_C}}{2}\right)-
     \\ &\max\left\{h\left(\frac{1-\braket{Z_AZ_B}}{2}\right), h\left(\frac{1-\braket{Z_AZ_C}}{2}\right)\right\}
\end{align*} 
where $h(x)$ is the binary entropy \cite{Grasselli2018}.
\\
Here we use the shorthand notation for Pauli strings e.g. $X_AX_BX_C=\sigma_x\otimes \sigma_x\otimes \sigma_x$.
To derive an upper bound for the asymptotic conference key-rate we can make use of the inflation technique, a tool that allows us to derive restrictions based on the network structure \cite{infaltion1, inflation2}. It was shown in \cite{infaltion1} that in a PEN with 3 nodes that the constraint
\begin{align*}
    |\braket{Z_AZ_B}|+ |\braket{Z_AZ_C}|\le 1+|\braket{Z_B}|\ |\braket{Z_C}|
\end{align*}
has to be met. 
Together with the uncertainty relation (see \cite{guehne2005}) $\braket{X_AX_BX_C}^2+\braket{Z_i}\braket{Z_j}\le 1$ with $i,j \in \{A,B,C\}$ we obtain

\begin{align*}
    \braket{X_AX_BX_C}^2+2 \min\{\braket{Z_AZ_B}, \braket{Z_AZ_C} \} \le 2.
\end{align*}
This is an improved condition of \cite{Hansenne2022} for the preparability of a state in a network with bipartite source states.
With that we get the following upper bound for the asymptotic key-rate
\begin{align*}
    r_{\text{BB84}}&\le 1-h\left( \frac{1}{4}\right) \approx 0.188,
\end{align*}
which turns out to be a tight bound. For instance, the state 
\begin{equation*}
    \rho\!=\!\frac{1}{2}\!\left( \ket{\phi^+} \bra{\phi^+}_{AB}\!\otimes \ket{+} \bra{+}_C+ \ket{\phi^+}\bra{\phi^+}_{AC}\!\otimes \ket{+} \bra{+}_B\right)
\end{equation*}
gives rise to a key-rate of $1-h\left( \frac{1}{4}\right)$ (see \cite{Carrara2021}).

\section{Proof of Theorem~\ref{PropDasCut}}
\label{SecDasCutProof}

The proof follows (generalizes) one of the proofs that the relative entropy of (bipartite) entanglement for pure states is equal to the entanglement entropy given in Ref.~\cite{VedralPlenio98}. But a difference with respect to the bipartite case is that we need to consider different cuts of a multipartite entangled state into a product of two parts. Our generalization of the proof from Ref.~\cite{VedralPlenio98} works only for the case when the bipartite pure states $\ket{\psi_e}$ in Eq.~(\ref{EqPENpure}) are maximally entangled. So we start with this case and return to the general case later. The case of the general $\ket{\psi_e}$ can be reduced to the case of the maximally entangled $\ket{\psi_e}$ by the entanglement purification.

\begin{lemma}
\label{LemDasCut}
    The statement of Theorem~\ref{PropDasCut} is true if $\ket{\psi_e}$ in Eq.~(\ref{EqPENpure}) are maximally entangled on the corresponding Hilbert spaces $\mathcal H_e$:
    \begin{equation}
    \label{EqPsieSchmidt}
        \ket{\psi_e}=\ket{\Phi_e}:=
        \frac{1}{\sqrt{d_e}}
        \sum_{n_e=1}^{d_e}
        \ket{n_e}\ket{n_e}
    \end{equation}
    for some bases $\ket{n_e}$ in the corresponding Hilbert spaces.
\end{lemma}

That is, the basis $\ket{n_e}$ in each Hilbert space is not the computational basis, but the basis specified by the Schmidt decomposition of $\ket{\psi_e}$. For short, we will also write $\ket{n_e}\ket{n_e}=\ket{n_en_e}$.

\begin{proof}
Consider two arbitrary (generally mixed) states $\sigma_{\vec A}$ and $\sigma_{\vec A}^*$ and the function
\begin{equation}
\label{EqfD}
    f(x)=(\ln 2) \,
    D(\rho_{\vec A}
    \, \| \,
    (1-x)\sigma^*_{\vec A}
    +x\sigma_{\vec A}).
\end{equation}
Here $\sigma_{\vec A}^*$ will be our guess for the minimum of the relative entropy, i.e. is treated as a fixed biseparable state and $\sigma_{\vec A}$ is a variable biseparable function. We will be interested in $f'(0)$, which is thus the directional derivative of the quantum relative entropy (up to a constant factor or, equivalently, base of the logarithm) with respect to $\sigma_{\vec A}$ in the point $\sigma^*_{\vec A}$.
For simplicity, we will drop the subindex $\vec A$. 

Using the representation
\begin{equation}
    \ln a=\int_0^\infty
    \frac{at-1}{a+t}
    \cdot
    \frac{dt}{1+t^2},
\end{equation}
it can be shown \cite{VedralPlenio98,AudenaertEisert2011} that the derivative of the function $f(x)$ in the point $x=0$ is 
\begin{equation}
\label{Eqfderiv}
    f'(0)=1-\int_0^\infty
    \Tr\left(
    \frac{1}{\sigma^*+t}
    \,\rho\,
    \frac{1}{\sigma^*+t}
    \,\sigma
    \right)dt
\end{equation}
for $x\in[0,1]$.

Set
\begin{equation}
\label{EqSigmaStar}
    \sigma^*=
    \bigotimes_{e\in C^*}
    \sigma_e^*
    \otimes
    \bigotimes_{e\notin C^*}
    \ket{\psi_e}\bra{\psi_e}
\end{equation}
for some cut $C^*$ (to be fixed later), where
\begin{equation}
\label{EqSigmaStare}
    \sigma^*_e=
    \frac{1}{d_e}
    \sum_{n_e=1}^{d_e}
    \ket{n_en_e}
    \bra{n_en_e}.
\end{equation}
$\sigma^*$ is our guess for minimum for a proper $C^*$. It is straightforward to calculate that 
\begin{equation}
\label{EqDS}
    D(\rho \, \| \, \sigma^*)=
    \sum_{e\in C^*}S_e=
    \sum_{e\in C^*}\log d_e.
\end{equation}
Also,
\begin{equation}
\label{EqSigmaStarFrac}
\begin{split}
    \frac{1}{\sigma^*+t}
    =
    \frac{1}{
    \prod\limits_{e\in C^*}
    d_e^{-1}+t
    }
    &\bigotimes_{e\in C^*}
    \left(
    \sum_{n_e=1}^{d_e}
    \ket{n_en_e}
    \bra{n_en_e}
    \right)
    \\
    \otimes
    &\bigotimes_{e\notin C^*}
    \ket{\Phi_e}\bra{\Phi_e}
\end{split}
\end{equation}
and
\begin{multline}
\label{EqSigmaStarFracs}
    \frac{1}{\sigma^*+t}
    \,\rho\,
    \frac{1}{\sigma^*+t}
    =
    \frac{
    \prod\limits_{e\in C^*}
    d_e^{-1}
    }{
    \big(\prod\limits_{e\in C^*}
    d_e^{-1}+t
    \big)^2
    }
    \\
    \times\bigotimes_{e\in C^*}
    \left(\sum_{n_e,m_e=1}^{d_e}
    \ket{n_en_e}
    \bra{m_em_e}
    \right)
    \otimes
    \bigotimes_{e\notin C^*}
    \ket{\Phi_e}\bra{\Phi_e}
\end{multline}

Let now $\sigma=\ket\varphi\bra\varphi$ be a pure biseparable state. This means that it is separable with respect to a certain vertex partition. Denote the corresponding cut as $C$. Then $\ket\varphi$ can be written as
\begin{equation}
\label{Eqphi}
    \ket\varphi=
    \bigotimes_{e\in C}
    \ket{\alpha_e}\ket{\beta_e}
    \otimes
    \bigotimes_{e\notin C}
    \ket{\varphi_e},
\end{equation}
where $\ket{\varphi_e}$ are arbitrary (generally, entangled) bipartite states from the Hilbert spaces $\mathcal H_e$, and $\ket{\alpha_e}\ket{\beta_e}$ are pure separable states on $\mathcal H_e$:
\begin{equation}
\label{EqAlphaBeta}
    \ket{\alpha_e}=\sum_{n_e=1}^{d_e}
    a_{n_e}^e\ket{n_e},
    \qquad
    \ket{\beta_e}=\sum_{n_e=1}^{d_e}
    b_{n_e}^e\ket{n_e}.
\end{equation}

Substitution of Eqs.~(\ref{EqSigmaStarFracs}), (\ref{Eqphi}), and (\ref{EqAlphaBeta}) into Eq.~(\ref{Eqfderiv}) gives
{}
\begin{equation}
\label{Eqfderiv4}
\begin{split}
    1-f'(0)
    =
    &\bigotimes_{e\in C^*\cap C}
    \sum_{n_e,m_e=1}^{d_e}
    \overline{a_{n_e}^e
    b_{n_e}^e}\,
    a_{m_e}^e b_{m_e}^e
    \\
    \otimes
    &\bigotimes_{e\in C^*\backslash C}
    \sum_{n_e,m_e=1}^{d_e}
    \braket{m_em_e|\varphi_e}
    \braket{\varphi_e|n_en_e}
    \\
    \otimes
    &\bigotimes_{e\in C\backslash C^*}
    \frac{1}{d_e}
    \sum_{n_e,m_e=1}^{d_e}
    \overline{a_{n_e}^e
    b_{n_e}^e}\,
    a_{m_e}^e b_{m_e}^e
    \\
    \otimes
    &\bigotimes_{e\notin C\cup C^*}
    |\braket{\Phi_e|\varphi_e}|^2,
\end{split}
\end{equation}
where we have taken into account that $\int_0^\infty\frac{c\,dt}{(c+t)^{2}}=1$. 

Now estimate the absolute values of the all four groups of terms in Eq.~(\ref{Eqfderiv4}). For the fourth line we apply the bound $|\braket{\Phi_e|\varphi_e}|\leq1$. For the first and the third lines we apply the Cauchy-Bunyakovsky-Schwarz inequality:
{}
\begin{equation}
\begin{split}
    \left|
    \sum_{n_e,m_e=1}^{d_e}
    \overline{
    a_{n_e}^e
    b_{n_e}^e}\,
    a_{m_e}^e b_{m_e}^e
    \right|
    &\leq
    \left(
    \sum_{n_e=1}^{d_e}
    |a_{n_e}^eb_{n_e}^e|
    \right)^2
    \\
    &\leq
    \left(
    \sum_{n_e=1}^{d_e}
    |a_{n_e}^e|^2
    \right)
    \left(
    \sum_{n_e=1}^{d_e}
    |b_{n_e}^e|^2
    \right)
    \\&=1.
\end{split}
\end{equation}
For the second line of Eq.~(\ref{Eqfderiv4}), we consider the maximization problem
\begin{equation}
    \sum_{n_e=1}^{d_e}
    |\braket{n_en_e|\varphi_e}|\to\max
\end{equation}
such that
\begin{equation}
    \sum_{n_e=1}^{d_e}
    |\braket{n_en_e|\varphi_e}|^2\leq1.
\end{equation}
Its solution is $|\braket{n_en_e|\varphi_e}|=1/\sqrt{d_e}$ giving
\begin{equation}
\begin{split}
    \sum_{n_e,m_e=1}^{d_e}
    |\braket{m_em_e|\varphi_e}
    \braket{\varphi_e|n_en_e}|
    &=\left(
    \sum_{n_e=1}^{d_e}
    |\braket{n_en_e|\varphi_e}|
    \right)^2
    \\&\leq
    d_e.
\end{split}
\end{equation}
Putting it all together, we obtain
\begin{equation}
    |1-f'(0)|\leq
    \frac{
    \prod\limits_{
    e\in C^*\backslash C}
    d_e
    }
    {
    \prod\limits_{
    e\in C\backslash C^*}
    d_e
    }
    =
    \frac{
    \prod\limits_{
    e\in C^*}
    d_e
    }
    {
    \prod\limits_{
    e\in C}
    d_e
    }.
\end{equation}
Now let us choose $C^*$ such that
\begin{equation}
    \prod_{e\in C^*}d_e=
    \min_{\text{cuts }C}
    \prod_{e\in C}d_e,
\end{equation}
or, equivalently,
\begin{equation}
    \sum_{e\in C^*}\log d_e=
    \min_{\text{cuts }C}
    \sum_{e\in C}\log d_e.
\end{equation}
Then $|1-f'(0)|\leq1$ and, hence, $f'(0)\geq0$ for all pure biseparable $\sigma$. Since $f'(0)$ is linear in $\sigma$ and every biseparable state is a convex combination of pure biseparable states, the inequality $f'(0)\geq0$ holds for all biseparable $\sigma$. Together with the convexity of the quantum relative entropy in Eq.~(\ref{EqfD}), this means that $\sigma^*$ provides minimum to the quantum relative entropy over all biseparable states. Together with Eq.~(\ref{EqDS}), this means that the statement of Proposition~\ref{PropDasCut} is true if $\ket{\psi_e}$ are maximally entangled states (\ref{EqPsieSchmidt}).

\end{proof}

Consider now the case of a general $\ket{\psi_e}$. Consider again their Schmidt decomposition: 
    \begin{equation}
        \ket{\psi_e}=
        \sum_{n_e=1}^{d_e}
        \sqrt{p_{n_e}^e}
        \ket{n_en_e}.
    \end{equation}
As a guess for the minimum of the relative entropy, we again consider the state $\sigma^*$ (\ref{EqSigmaStar}), where
{}
\begin{equation}
    \sigma^*_e=
    \sum_{n_e=1}^{d_e}
    \sqrt{p_{n_e}^e}
    \ket{n_en_e}
    \bra{n_en_e}
\end{equation}
and $C^*$ is chosen such that
\begin{equation}
    \sum_{e\in C^*}
    S_e
    =
    \min_{\text{cuts }C}
    \sum_{e\in C}
    S_e.
\end{equation}
The direct calculation shows that 
\begin{equation}
\label{EqDSgen}
    D(\rho_{\vec A} \, \| \, \sigma^*_{\vec A})=\sum_{e\in C^*}S_e.    
\end{equation}
So, we need to show that the relative entropy of entanglement does not exceed the right-hand side of Eq.~(\ref{EqDSgen}).

As is well-known \cite{ChristandlThesis,WildeBook}, the distillable (bipartite) entanglement of the pure state $\ket{\psi_e}$ is equal to its entanglement entropy $S_e$. This means that, for every $\varepsilon>0$ and every $\delta>0$, there exists a family of LOCC maps $\Lambda_e^{(n)}$ acting on $n$ copies of $\ket{\psi_e}\bra{\psi_e}$ such that 
{}
\begin{equation}
    \|
    \Lambda_e^{(n)}(\ket{\psi_e}\bra{\psi_e}^{\otimes n})-\ket{\Phi_e}\bra{\Phi_e}^{
    \otimes 
    \lfloor n(S_e-\delta)\rfloor
    }\|\leq
    \frac{\varepsilon}{
    |\mathcal E|
    }
\end{equation}
for all sufficiently large $n$. Composition of $\Lambda_e^{(n)}$ for all $e$ constitutes the map $\Lambda^{(n)}$ acting on $n$ copies of the PEN state $\rho_{\vec A}$, see Eq.~(\ref{EqPENpure}). Then
{}
\begin{equation}
    \|
    \Lambda{(n)}(
    \rho_{\vec A}^{\otimes n}
    )-
    \tau_{\vec A}^{(n)}
    \|\leq\varepsilon,
\end{equation}
where
\begin{equation}
\label{EqPENpuremax}
    \tau_{\vec A}^{(n)}=
    \bigotimes_{e\in\mathcal E}
    \ket{\Phi_e}
    \bra{\Phi_e}^{
    \otimes 
    \lfloor n(S_e-\delta)\rfloor
    }.    
\end{equation}

We have the following chain:
\begin{equation}
\label{EqGMEchain}
\begin{split}
    E_{\rm GME}(\rho_{\vec A})
    &\geq
    \frac1n
    E_{\rm GME}(
    \rho_{\vec A}^{\otimes n}
    )
    \\
    &\geq
    \frac1n
    E_{\rm GME}\big(
    \Lambda^{(n)}(\rho_{\vec A}^{\otimes n})
    \big)
    \\
    &\geq
    \frac1n
    [E_{\rm GME}(
    \tau_{\vec A}^{(n)}
    )-g(\varepsilon,n)]
    \\&=
    \frac1n
    \left[
    \min_C\sum_{e\in C}
    \lfloor n(S_e-\delta)\rfloor -g(\varepsilon,n)
    \right]
    \\
    &\geq\min_C
    \sum_{e\in C}
    \Big(S_e-\delta-\frac1n\Big)
    -\frac{g(\varepsilon,n)}n,
\end{split}
\end{equation}
where 
\begin{equation}
    g(\varepsilon,n)=
    2\varepsilon
    \log\left(
    \prod_{e\in\mathcal E}
    d_e^{
    \lfloor n(S_e-\delta)\rfloor
    }
    \right)
    -2\varepsilon\log\varepsilon
    +4\varepsilon.
\end{equation}
The first inequality in chain (\ref{EqGMEchain}) follows from the definition of the relative entropy of GME (it is the standard regularization procedure). In the second inequality, we have used the monotonicity of the relative entropy of GME $E_{\rm GME}$ under LOCC maps.
The third inequality is the continuity bound for such kind of quantities (relative entropy with respect to a certain class of density operators) \cite{DonaldHoro99}. In the fourth line, we have used the proved statement (Lemma~\ref{LemDasCut}) about the relative entropy of GME for PEN states of form (\ref{EqPENpuremax}) and the minimum is taken over cuts $C$. The last inequality is due to the rounding operation.

Since $\delta$, $\varepsilon$, $g(\varepsilon,n)/n$, and $1/n$ can be made arbitrarily small, we conclude that 
\begin{equation}
    E_{\rm GME}(\rho_{\vec A})\geq
    \min_C \sum_{e\in C}S_e.
\end{equation}
Together with Eq.~(\ref{EqDSgen}), this gives the required statement.

\section{Proof of Theorem~\ref{ThPENpure}}
\label{proofThPENpure}

Since we consider pure $\rho_{\vec A}$, there is now Eve's register $E$ in Eq.~(\ref{EqKeyRate}); the only Eve's knowledge is classical communication (and shared randomness), i.e., register $C$. Consider first the case of no shared randomness, i.e., no $\lambda$. For an arbitrary partition $\mathcal P$ (\ref{EqP}), denote the maximal total correlation of the outcomes of local measurements 
\begin{equation}
\label{EqMultIacc}
\begin{split}
    I_{\rm loc}(\mathcal P)&\equiv
    I_{\rm loc}(A_{J_1}:
    A_{J_2}:
    \ldots
    \colon
    A_{J_p})\\&=
    \max 
    I(X_{J_1}:
    X_{J_2}:\ldots:
    X_{J_p}),
\end{split}
\end{equation}
where the maximum is taken over $\mathcal P$-local POVMs (i.e., collective measurements inside the subsets $J_\alpha\in\mathcal P$ are allowed) and $X_{J_\alpha}$, $\alpha=1,\ldots,p$, denote the random variables corresponding to the measurement outcomes. Here
\begin{equation}
    I(X_{J_1}:\ldots:
    X_{J_p})=
    \sum_{\alpha=1}^p
    H(X_{J_1})
    -H(X_{J_1}\ldots X_{J_p})\geq0
\end{equation}
is the total correlation (a multipartite generalization of the mutual information) and $H$ denotes the Shannon entropy of a random variable.

If we include shared randomness $\lambda$, we will be interested in the corresponding conditional quantity conditioned on $\lambda$:
\begin{equation}
\label{EqMultIaccL}
\begin{split}
    I_{\rm loc}(\mathcal P|L)&\equiv
    I_{\rm loc}(A_{J_1}:
    A_{J_2}:\ldots:
    A_{J_p}|L)\\&=
    \max 
    I(X_{J_1}:
    X_{J_2}:\ldots:
    X_{J_p}|L),
\end{split}
\end{equation}
where the maximum is taken over all $\mathcal P$-local measurements with shared randomness and $L$ is the shared randomness as a random variable (in contrast to $\lambda$, which denotes a concrete value of it). Here,
\begin{multline}    
\label{EqMultIC}
    \!I(X_{J_1}:\!
    \ldots:
    \!X_{J_p}|C)\!=\!\sum_{\alpha=1}^p
    H(X_{J_1}|C)
    -H(X_{J_1}\!\ldots X_{J_p}|C)\\
    =
    \sum_{c\in\mathcal C}
    P(C=c)
    I(X_{J_1}\colon
    \ldots
    \colon
    \!X_{J_p}|C=c)
\end{multline}
is the conditional total correlation for an arbitrary random variable $C$ taking values on the set $\mathcal C$. In our case this is classical communication, which includes the public randomness. Due to the last line in Eq.~(\ref{EqMultIC}), the maximization in Eq.~(\ref{EqMultIaccL}) is reduced to the maximization over $\mathcal P$-local POVMs without shared randomness. In other words, shared randomness does not give advantage for the maximization of information. So, in the following, we will consider the maximal total correlation (\ref{EqMultIacc}) without shared randomness.

Consider first the case of two parties $A_1=A$ (``Alice'') and $A_2=B$ (``Bob'') and a pure bipartite state $\ket{\psi_{AB}}$. Denote $\rho_A=\Tr_B\ket{\psi_{AB}}\bra{\psi_{AB}}$ and $\rho_B=\Tr_A\ket{\psi_{AB}}\bra{\psi_{AB}}$. Denote also $S(\varrho)$ the von Neumann entropy of an arbitrary density operator $\varrho$.

The fact from the following lemma is well-known since it is actually a reformulation of the statement that the accessible information for an ensemble of quantum states cannot be greater than the Holevo information of this ensemble \cite{HolevoBook}. 
\begin{lemma}
\label{LemMLDbipartite}
Let $\ket{\psi_{AB}}$ be a pure bipartite state. Then
\begin{equation}
\label{EqBipartitePiani}
    I_{\rm loc}(A:B)=S(\rho_A),
\end{equation}
\end{lemma}

\begin{proof}
    The value $S(\rho_A)$ for the maximal mutual information is obviously achieved when both $A$ and $B$ measure in the bases corresponding to the Schmidt decomposition of $\ket{\psi_{AB}}$. So, it suffices to prove that a higher value cannot be achieved. 

    We need to maximize over Alice's and Bob's local POVMs $\{M_x^A\}_{x\in\mathcal X}$ and $\{M_y^B\}_{y\in\mathcal Y}$. Denote $X$ and $Y$ the random variables corresponding to their outcomes. We can express the maximal mutual information as a sequential maximization:
    \begin{equation}
        I_{\rm loc}(A:B)
        =\max_{\{M_x^A\}}
        \max_{\{M_y^B\}}
         I(X:Y),
    \end{equation}
    where Bob's POVM can depend on Alice's one. In other words, for each Alice's POVM, we optimize over Bob's POVMs and obtain a function of Alice's POVM. Then we optimize it over Alice's POVMs.
    
    Let $p_x$ be the probabilities of Alice's outcomes and $\rho_B^{(x)}$ be the corresponding postmeasurement states of Bob,   
    so that $\sum_xp_x\rho_B^{(x)}=\rho_B$. Hence, $I(X:Y)$ maximized over Bob's POVM for a fixed Alice's POVM is nothing else as the accessible information for the ensemble $\{p_x,\rho_B^{(x)}\}_{x\in\mathcal X}$ \cite{HolevoBook}. It is upper bounded by the Holevo quantity, which, it turn, is upper bounded by
    the entropy of the average state of the ensemble $\rho_B$:
    \begin{equation}
        \max_{\{M_y^B\}_y}
         I(X:Y)
         \leq
         S(\rho_B)-
         \sum_{x\in\mathcal X}
         p_x
         S(\rho_B^{(x)})
         \leq
         S(\rho_B).
    \end{equation}
    Since this inequality is true for any Alice's POVM, we conclude that also
    \begin{equation}
        I_{\rm loc}(X:Y)
        \leq
        S(\rho_B)=S(\rho_A).
    \end{equation}
\end{proof}

\begin{corollary}
\label{CorBipartitePianiGen}
    Consider a more general bipartite state:
    \begin{equation}
    \label{EqBipartGen}
        \rho_{AB}
        =\ket{\psi_{A'B'}}
        \bra{\psi_{A'B'}}
        \otimes
        \rho_{A''}\rho_{B''},
    \end{equation}
    i.e., besides the pure (possibly entangled) bipartite state $\ket{\psi_{A'B'}}$, Alice and Bob have their own ``private'' uncoupled states $\rho_{A''}$ and $\rho_{B''}$, respectively, so that $A=(A',A'')$ and $B=(B',B'')$. Then
\begin{equation}
\label{EqBipartitePianiGen}
    I_{\rm loc}(A:B)=S(\rho_{A'}),
\end{equation}
where, as before, $\rho$ with subindices means the partial trace of $\rho_{AB}\equiv\rho_{A'A''B'B''}$ over the rest subsystems. 
\end{corollary}
That is, additional ``private'' subsystems uncoupled from the pure bipartite state do not change the left-hand side of Eq.~(\ref{EqBipartitePianiGen}).
\begin{proof}
    The right-hand side of Ineq.~(\ref{EqBipartitePianiGen}) is achieved when Alice and Bob again measure $A'$ and $B'$ in the bases corresponding to the Schmidt decomposition of $\ket{\psi_{A'B'}}$ and perform the trivial measurements (i.e., corresponding to the identity operators) on $A''$ and $B''$. The fact that this bound cannot be overcome follows from Lemma~\ref{LemMLDbipartite}. Indeed, preparation of local subsystems $A''$ and $B''$ uncoupled from the entangled state $\psi_{A'B'}$ and local joint measurements on $A'A''$ and $B'B''$ falls into the general concept of local POVMs on $A'$ and $B'$.
\end{proof}

\begin{lemma}
\label{LemMultIacc}
If $\ket{\psi_{\vec A}}=\bigotimes_{e\in\mathcal E}\ket{\psi_e}$ is a pure PEN state and $\mathcal P$ is a partition of $[N]$ (\ref{EqP}), then
    \begin{equation}
        I_{\rm loc}(A_{J_1}:\ldots :A_{J_p})
        \leq\sum_{e\in\mathcal E(\mathcal P)}S_e.
    \end{equation}
\end{lemma}

\begin{proof}
    Let us apply the following decomposition of the multipartite mutual information:

    \begin{equation}
    \label{EqIdecompose}
    \begin{split}
        I(X_{J_1}:X_{J_2}:\ldots:
        X_{J_{p}})
        &=I(X_{J_1}:X_{J_2})
        \\&+I(X_{J_3}:X_{J_1}X_{J_2})
        +\ldots
        \\
        &+I(X_{J_p}:X_{J_1}\ldots X_{J_{p-1}}).
    \end{split}
    \end{equation}
    Thus, we have reduced the multipartite mutual information to a sum of bipartite ones. 
    
    Consider the first term $I(J_1:J_2)$. The corresponding quantum state $\rho_{A_{J_1\cup J_2}}$ has the form~(\ref{EqBipartGen}):
    \begin{equation}
        \rho_{A_{J_1\cup J_2}}=
        \bigotimes_{e\in\mathcal E(J_1,J_2)}
        \ket{\psi_e}\bra{\psi_e}
        \otimes
        \rho^{(J_1,\overline J_2)}
        \otimes
        \rho^{(J_2,\overline J_1)},
    \end{equation}
    where $\mathcal E(J_1,J_2)\subset\mathcal E$ denotes the set of edges whose one endvertex belongs to $J_1$ and the other one belongs to $J_2$ and $\rho^{(J_1,\overline J_2)}$ (analogously $\rho^{(J_2,\overline J_1)}$) denotes a state originating from vertices connecting $J_1$ to vertices from $\overline J_2=[N]\backslash J_2$. That is, we have pure states $\ket{\psi_e}$ for the edges connecting vertices inside the subset $J_1$ and the partial traces of $\ket{\psi_e}\bra{\psi_e}$ for the edges connecting $J_1$ with vertices not from $J_1$ or $J_2$. Hence, in force of Corollary~\ref{CorBipartitePianiGen},
    \begin{equation}
        \max I(J_1:J_2)=\sum_{e\in\mathcal E(J_1,J_2)}S_e
    \end{equation}
    and the maximum is achieved on the measurements of $\ket{\psi_e}$, $e\in\mathcal E(J_1,J_2)$, corresponding to their Schmidt decompositions.

    Analogously,
    \begin{equation*}
    \begin{split}
        \max I(J_3:J_1J_2)&= 
        \sum_{e\in\mathcal E(J_1\cup J_2,J_3)}S_e,\\
        &\ldots\\
        \max I(J_{|\mathcal P|}:J_1\ldots J_{|\mathcal P|-1})&=
        \sum_{e\in\mathcal E(J_1\cup\ldots\cup J_{|\mathcal P|-1},J_{|\mathcal P|})}S_e,
    \end{split}
    \end{equation*}
    where the maxima are achieved on the measurements of $\ket{\psi_e}$ (for disjoint subsets of edges) corresponding to their Schmidt decompositions. Since $\mathcal E(J_1,J_2)$, $\mathcal E(J_1\cup J_2,J_3)$, \ldots, and $\mathcal E(J_1\cup\ldots\cup J_{|\mathcal P|-1},J_{|\mathcal P|})$ constitute a decomposition of $\mathcal E$,
    \begin{equation}
        \max
        I(X_{J_1}:X_{J_2}:\ldots:
        X_{J_{|\mathcal P|}})
        \leq 
        \sum_{e\in\mathcal E}
        S_e.
    \end{equation}
    Since the measurement achieving the maximums for the terms in Eq.~(\ref{EqIdecompose}) are compatible (act nontrivially on disjoint subsets of edges), the upper bound $\sum_{e\in\mathcal E} S_e$ is also achievable (on the measurements corresponding to the Schmidt decompositions of all $\ket{\psi_e}$).    
\end{proof}

\begin{lemma}
\label{LemIcond}
    Consider arbitrary random variables $X_1$, \ldots, $X_p$ and a random variable $C$, which is a function of, e.g., $X_1$ and private local randomness of the user~1. Then
    \begin{equation}
        I(X_1:\ldots:X_p|C)
        \leq
        I(X_1:\ldots:X_p).
    \end{equation}
\end{lemma}

\begin{proof}
    We have
    \begin{equation}
    \label{EqIX1XpC}
        I(X_1:\ldots:X_p|C)
        =\sum_{\alpha=1}^p H(X_\alpha|C)-
        H(X_1\ldots X_p|C).
    \end{equation}
    We can write
    \begin{equation}
    \label{EqHX1C}
    \begin{split}
        H(X_1|C)&=H(X_1C)-H(C)\\&=H(X_1)+H(C|X_1)-H(C).
    \end{split}
    \end{equation}
    Analogously,
    \begin{equation}
    \label{EqHX1XpC}
    \begin{split}
        H(&X_1\ldots X_p|C)\\&=H(X_1\ldots X_p)+H(C|X_1\ldots X_p)-H(C)\\&=H(X_1\ldots X_p)+H(C|X_1)-H(C),
    \end{split}
    \end{equation}
    where we have used that $X_2\ldots X_p$ do not contain additional information on $C$ with respect to $X_1$, hence, $H(C|X_1\ldots X_p)=H(C|X_1)$.

    Substitution of Eqs.~(\ref{EqHX1C}) and (\ref{EqHX1XpC}) as well as the inequalities $H(X_\alpha|C)\leq H(X_\alpha)$ for $\alpha=2,\ldots,p$ into Eq.~(\ref{EqIX1XpC}) gives the statement of the lemma.
\end{proof}

\begin{corollary}
\label{CorIcond}
    Consider arbitrary random variables $X_1$, \ldots, $X_p$, and $C$, and a random variable $C'$, which is a function of, e.g., $X_1$, $C$, and private local randomness of the user~1. Then
    \begin{equation}
        I(X_1:\ldots:X_p|CC')
        \leq
        I(X_1:\ldots:X_p|C).
    \end{equation}
\end{corollary}
Such generalization is straightforward: We simply replace $C$ by $C'$ and add conditioning on $C$ in all calculations in the proof of Lemma~\ref{LemIcond}.

Lemma~\ref{LemIcond} and Corollary~\ref{CorIcond} say that the classical multipartite information cannot be increased by local calculations and public communication.

\begin{proof}[Proof of Theorem~\ref{ThPENpure}]

    Consider an apbitrary $I$-proper partition $\mathcal P$ and consider the ideal state (\ref{EqKeyIdeal}). Consider the subsets $J_1,\ldots,J_p\in\mathcal P$ as ``aggregated'' users (i.e., merge users from the same subset together). Then, since each $J_\alpha$ intersects with $I$, the corresponding aggregated user has its own copy of the key. Denote the corresponding random variable $K_{J_\alpha}$. It is straightforward to calculate that
    \begin{equation}
        I(K_{J_1}
        \colon
        K_{J_2}
        \colon
        \ldots
        \colon
        K_{J_p})=(p-1)m.
    \end{equation}
    
    From the other side, from Lemmas~\ref{LemMultIacc}, \ref{LemIcond} and Corollary~\ref{CorIcond}, we conclude that the multipartite information between measurement outcomes $X_{J_1}$, \ldots, $X_{J_p}$ followed by any classical postprocessing with public communication, cannot exceed the right-hand side of (\ref{EqPENpure}), which proves the theorem.
\end{proof}

\section{Proof of Theorem~\ref{ThPEN}}
\label{proofThPEN}

Consider an arbitrary decomposition (\ref{EqRhoABdecompose}). So, Alice's and Bob's density $\rho_{AB}$ can be obtained as  $\Tr_E\ket\psi_{ABE}\!\bra{\psi}$ for the following tripartite state $\ket\psi_{ABE}$ with Eve:
\begin{equation}
    \ket\psi_{ABE}
    =\sum_x\sqrt{p_x}\ket{\psi_x}_{AB}\ket{x}_E,
\end{equation}
where $\ket{x}_E$ are orthonormal vectors in Eve's space.
Suppose that Eve measures her register in the standard basis and announces the result $x$. Of course, the fact of announcement is advantageous for Alice and Bob. Then they know that they share a pure state $\ket{\psi_x}$ in this position. After a large number $n$ rounds, Alice and Bob have a sequence of $\ket{\psi_x}$, where, with a high probability, the number of occurrences of each $x$ is approximately $np_x$. Since the distillable key is additive for a pure states and is equal to the entanglement entropy, we arrive at the conclusion that  distillable key is the entanglement entropy averaged over the ensemble $\{p_\alpha,\ket{\psi_\alpha}\}$. This is an upper bound since it corresponds to a particular Eve's attack (purification) with the announcement of Eve's outcomes to Alice and Bob, which is, of course, advantageous for them.

Minimization over decompositions (\ref{EqRhoABdecompose}) gives $E_{\rm F}(\rho_{AB})$ as an upper bound. The same arguments applied to $n$ rounds give the upper bound $E_{\rm F}(\rho_{AB}^{\otimes n})/n$. Taking the limit $n\to\infty$ gives the regularized entanglement entropy as an upper bound for the distillable key in the bipartite scenario. This idea (for the bipartite case) was suggested in Ref.~\cite{KoashiWinter2004}. We give a detailed rigorous proof.

For simplicity of notations, let us first prove the known statement that the bipartite distillable key is upper bounded by the regularized entanglement of formation based on this intuition. Then we will generalize this proof for the multipartite case.

\begin{lemma}
\label{LemBiKeyEF}
    Let two parties $A$ (Alice) and $B$ (Bob) have a source of bipartite state $\rho_{AB}$ acting on a Hilbert space $\mathcal H_{AB}=\mathcal H_A\otimes\mathcal H_B$ want to establish a secret key. Then the distillable key is upper bounded by $E_F^\infty(\rho_{AB})$.
\end{lemma}

\begin{proof}
    Consider an arbitrary decomposition
    \begin{equation}
        \rho_{AB}^{\otimes n}=\sum_{x\in\mathcal X}p_{x}\ket{\psi_{x}}\bra{\psi_{x}}
    \end{equation}
    for some $n>0$ and a purification
    \begin{equation}
        \ket{\Psi_{ABE}}^{\otimes n}=\sum_{x\in\mathcal X}
        \sqrt{p_{x}}
        \ket{\psi_{x}}_{(AB)^n}
        \ket{x}_{E^n},
    \end{equation}
    where $\{\ket{x}\}$ is a set of orthonormal vectors in $\mathcal H_E^{\otimes n}$.
    Without loss of generality, assume that all $p_x$ are positive. 
    
    For $M\geq1$, denote $\mathbf x=(x_1,\ldots,x_M)\in\mathcal X^M$ (we will also refer to $\mathbf x$ as sequences) and
    \begin{equation}
        p_{\mathbf x}=\prod_{i=1}^M p_{y_{i}},
        \quad
        \ket{\psi_{\mathbf x}}
        =\bigotimes_{i=1}^M
        \ket{\psi_{y_{i}}},
        \quad
        \ket{\mathbf x}
        =\bigotimes_{i=1}^M
        \ket{x_{i}},
    \end{equation}
    Then
    \begin{equation}
    \label{EqPsiABEnM}
        \ket{\Psi}_{ABE}^{\otimes nM}
        =
        \sum_{\mathbf x\in\mathcal X^m}
        \sqrt{p_{\mathbf x}}
        \ket{\psi_{\mathbf x}}
        \ket{\mathbf x}.
    \end{equation}

    Let us consider the distillable key (\ref{EqKeyRate}) for the state $\rho_{AB}^{\otimes n}$, i.e., a block of $n$ rounds we consider now as a one ``aggregated'' round:
    \begin{equation}
        \overline R=
        \lim_{\varepsilon\to0}
        \sup_{L,M,\Lambda}
        \Big\{\frac LM\,\Big|\,
        \frac12
        \|\Lambda(\rho_{ABE}^{\otimes nM})
        -\rho_{K_AK_BE^{nM}C}^{(L),\,\rm ideal}\|
        \leq\varepsilon
        \Big\}.
    \end{equation}
    We are going to prove that  
    \begin{equation}
    \label{EqMeanSbnd}
        \overline R\leq
        \overline S_{\rm ent}=
        \sum_{x\in\mathcal X} 
        p_x S_{\rm ent}(\ket{\psi_x}).    
    \end{equation}
    This means that, for any $\Delta>0$, there exists $\varepsilon>0$ such that, if $L/M>\overline S_{\rm ent}+\Delta$, then
    \begin{equation}
    \label{EqIneqEps0}
        \frac12
        \|\Lambda(\rho_{ABE}^{\otimes nM})
        -\rho_{K_AK_BEC}^{(L),\,\rm ideal}\|
        >\varepsilon_0
    \end{equation}
    for any LOSR+PP map $\Lambda$.
    
    Consider the channel $\mathcal D$ acting on Eve's subsystem which decoheres in the $\{\ket{e_{\mathbf x}}\}$ basis, or, in other words, performs a nonselective measurement in this basis. The map $\mathcal D$ commutes with $\Lambda_M$ (since they act nontrivially on different subsystems) and
    \begin{equation}
        \mathcal D
        \big(\rho_{ABE}^{\otimes nM}\big)
        =\widetilde\rho_{ABE}^{\otimes nM}
        =
        \sum_{\mathbf x\in\mathcal X^M}
        p_{\mathbf x}
        \ket{\psi_{\mathbf x}}
        \bra{\psi_{\mathbf x}}
        \otimes
        \ket{\mathbf x}
        \bra{\mathbf x}
    \end{equation}
    Obviously, the action of $\mathcal D$ on the ideal state gives the ideal state  $\widetilde\rho_{K_AK_BEC}^{(L_M),\,\rm ideal}$ corresponding to $\widetilde\rho_{ABE}^{\otimes nM}$, i.e., uniformly distributed key uncoupled from $E^n$ and $C$, and $E^n$ having the same marginal state as in $\rho_{ABE}^{\otimes nM}$:
    \begin{equation}
        \widetilde\rho_{K_AK_BEC}^{(L_M),\,\rm ideal}
        =
        \tau_{K_AK_B}
        \otimes
        \sum_{\mathbf x\in\mathcal X^M}
        p_{\mathbf x}
        \ket{\mathbf x}
        \bra{\mathbf x}
        \otimes
        \rho_{C|\mathbf x},
    \end{equation}
    where 
    \begin{equation}
        \tau_{K_AK_B}
        =2^{-m}\
        \sum_{k\in\{0,1\}^m}
        \ket{k,k}\bra{k,k}.
    \end{equation}
    
    The action of a CPTP map cannot increase the trace distance, hence, it suffices to prove the property with Ineq.~(\ref{EqIneqEps0}) for $\widetilde\rho_{ABE}$ rather than $\rho_{ABE}$ and the corresponding ideal state.

    Due to the block-diagonal structure of $\widetilde\rho_{ABE}^{\otimes nM}$ and $\widetilde\rho_{K_AK_BEC}^{(L_M),\,\rm ideal}$ with respect to the vectors $\ket{\mathbf x}$, it turns out that
    \begin{multline}
        \|\Lambda(
        \widetilde\rho_{ABE}^{\otimes nM}
        )
        -\widetilde\rho_{K_AK_BEC}^{(L),\,\rm ideal}\|
        \\=
        \sum_{\mathbf x\in\mathcal X^M}
        p_{\mathbf x}
        \big\|\Lambda(
        \ket{\psi_{\mathbf x}}
        \bra{\psi_{\mathbf x}}
        ^{\otimes M}
        )
        -
        \tau_{K_AK_B}
        \otimes
        \rho_{C|\mathbf x}
        \big\|.
    \end{multline}
    It turns out (see below) that it suffices to prove that, for all $\Delta>0$, there exists $\varepsilon_1>0$ and $\overline M>0$ such that, if  $L/M>\overline S_{\rm ent}+\Delta$ and $M\geq\overline M$, then
    \begin{equation}
    \label{EqIneqEps0x}
        \frac12
        \big\|\Lambda(
        \ket{\psi_{\mathbf x}}
        \bra{\psi_{\mathbf x}}
        ^{\otimes M}
        )
        -
        \tau_{K_AK_B}
        \otimes
        \rho_{C|\mathbf x}
        \big\|
        \geq\varepsilon_1
    \end{equation}
    for any LOSR+PP map $\Lambda$ and for any sequence $\mathbf x$ from a high-probability (typical) subset.

    Let us use the following definition of a typical sequence \cite{CsiszarKornerBook}: A sequence $\mathbf y$ is said to be typical (denoted as $\mathbf x\in T_M$) if
    \begin{equation}
        \left|\frac{M_x^{(\mathbf x)}}M -p_x\right|\leq\delta_M
    \end{equation}
    for all $x\in\mathcal X$,
    where $M_x^{(\mathbf x)}$ is the number of occurrences of the symbol $x$ in the sequence $\mathbf x$. As in Ref.~\cite{CsiszarKornerBook}, we will assume that $\delta_M\to0$ and $\sqrt M\delta_M\to\infty$ as $M\to\infty$. Then, the total probability of typical sequences is at least $1-\gamma_M$, where $\gamma_M\to0$ as $M\to\infty$. Also we assume that $\delta_M$ is non-increasing with $M$.

    If the property with Ineq.~(\ref{EqIneqEps0x}) is true for all typical sequences and $M\geq\overline M$, then Ineq.~(\ref{EqIneqEps0}) is also true with 
    \begin{equation}
        \varepsilon_0=(1-\gamma_{\overline M})\varepsilon_1
    \end{equation}
    for $M\geq \overline M$. Then, for $M<\overline M$, the left-hand side of Ineq.~(\ref{EqIneqEps0}) cannot be zero or even be made arbitrarily close to zero because it would lead to the left-hand side equal or arbitrarily close to zero for arbitrary large $M$ by concatenations of the shorter keys. Then, taking the minimum between $(1-\gamma_{\overline M})\varepsilon_1$ and the infimum of left-hand side of Ineq.~(\ref{EqIneqEps0}) for $M<\overline M$ gives $\varepsilon_0$ for which Ineq.~(\ref{EqIneqEps0}) is satisfied for all $M$. 
    
    In order to prove Ineq.~\eqref{EqIneqEps0x} for typical sequences, we need to gather the data of some length $M_0$ such that $\lceil M_0(p_x+\delta_{M_0})\rceil/M_0\approx p_x$. For this purpose, consider the following family of pure states depending on $M_0$:
    \begin{equation}
    \label{EqPsiT}
        \ket{\psi_T}=
        \bigotimes_{x\in\mathcal X}
        \ket{\psi_x}^
        {\otimes\lceil 
        M_0(p_x+\delta_{M_0})
        \rceil}
    \end{equation}
    (here $T$ is for ``typical'').
    This is a pure state, so its distillable key is equal to
    \begin{equation}
    \label{EqPsitSent}
    \begin{split}
        S_{\rm ent}(\ket{\psi_T})
        &=\sum_{x\in\mathcal X}
        \lceil 
        M_0(p_x+\delta_{M_0})
        \rceil
        S_{\rm ent}(\ket{\psi_x})
        \\&=M_0(\overline S_{\rm ent}
        +\eta_{M_0}),
    \end{split}
    \end{equation}
    where $\eta_{M_0}\to0$ as $M_0\to\infty$. We assume that $M_0$ is fixed and large enough to ensure $\eta_{M_0}<\Delta$. That is, there exists $\varepsilon_1>0$ such that, if $L'/M'>M_0(\overline S_{\rm ent}+\Delta)$, then
    \begin{equation}
    \label{EqPsiTineq}
        \frac12
        \|\Lambda'(\ket{\psi_T}
        \bra{\psi_T}^
        {\otimes M'})
        -\rho_{K_AK_BC}^{(L'),\,\rm ideal}\|
        >\varepsilon_1
    \end{equation}
    for any LOSR+PP map $\Lambda'$.
    
    Consider $M>M_0$ and define $M'=\lceil M/M_0\rceil$. For a typical sequence $\mathbf y$ of length $M$, we have
    \begin{equation}
    \label{EqPsiYMineq}
        M_y^{(\mathbf y)}\leq
        M(p_y+\delta_M)\leq
        M'\lceil M_0(p_y+\delta_{M_0})\rceil,
    \end{equation}
    so, the states $\ket{\psi_{\mathbf x}}$ for typical $\mathbf x$ can be obtained from $\ket{\psi_T}^{\otimes M'}$ by permutation and elimination of some subsystems by both parties simultaneously. 
    
    We are ready to prove Ineq.~(\ref{EqIneqEps0x})  with the same $\varepsilon_1=\varepsilon_1(\Delta)$ as from Ineq.~(\ref{EqPsiTineq}) for large enough $M$. Suppose that Ineq.~(\ref{EqIneqEps0x}) is not satisfied for some $L$, $M$, $\mathbf x\in T_M$, and $\Lambda$ and $L/M>\overline S_{\rm ent}+\Delta$. Then consider a  quantum channel $\Lambda'=\Lambda\circ\Phi_{\mathbf x}$ acting on $\ket{\psi_T}^{\otimes M'}$, where $\Phi_{\mathbf x}$ are local permutations and eliminations of subsystems (depending on $\mathbf x$ but not requiring even classical communication). For such choice of the channel $\Lambda'$, in the notations of Ineq.~(\ref{EqPsiTineq}), $L'=L$ and the left-hand sides of Ineqs.~(\ref{EqIneqEps0x}) and (\ref{EqPsiTineq}) coincide. Hence, the left-hand side of Ineq.~(\ref{EqPsiTineq}) is not bigger than $\varepsilon_1$. From the other side, from our assumption of violation of Ineq.~(\ref{EqIneqEps0x}), we have
    \begin{equation}
        \frac{L'}{M/M_0}
        >
        M_0(\overline S_{\rm ent}+\Delta).
    \end{equation}
    We have not yet obtained a contradiction with Ineq.~(\ref{EqPsiTineq}) because a condition for it is $L'/M'>M_0(\overline S_{\rm ent}+\Delta)$, but $M'\geq M/M_0$. However,
    \begin{equation}
    \label{EqLMM}
        \frac{L'}{M'}=
        \frac{L'}{M/M_0}
        -L'
        \frac{M'M_0-M}{MM'}.
    \end{equation}
    By construction, $M\geq M_0(M'-1)$, so, the last term in Eq.~(\ref{EqLMM}) is upper bounded by $(L'/M')(M_0/M)$ and, thus, is infinitesimal as $M\to\infty$. Then,
    \begin{equation}
        \frac{L'_{M'}}{M'}
        >
        M_0(\overline S_{\rm ent}+\Delta)
    \end{equation}
    for large enough $M$. Thus, Ineq.~(\ref{EqPsiTineq}) is violated for large enough $M$. Hence, Ineq.~(\ref{EqIneqEps0x}) also cannot be violated for large enough $M$. This finishes the proof of Ineq.~(\ref{EqMeanSbnd}).

    Minimization over decomposition~(\ref{EqRhoABdecompose}) gives $\overline R<E_{\rm F}(\rho_{AB}^{\otimes n})$. If we remember that we consider a block of $n$ rounds and return to the original definition of the distillable key (\ref{EqKeyRate}), we obtain
    \begin{equation}
        \overline r_{\rm key}\leq\frac 1n E_{\rm F}(\rho_{AB}^{\otimes n}).
    \end{equation}
    Taking the limit $n\to\infty$ gives that the bipartite distillable key cannot be larger than the regularized entropy of formation, or entanglement cost.
\end{proof}
    
\begin{proof}[Proof of Theorem~\ref{ThPEN}]

    A generalization to the multipartite case is straightforward. 
    For simplicity, consider first the case $I=[N]$ (i.e., a conference key is required for all parties) and the finest partition, i.e., each subset in the partition $\mathcal P$ consists from a single element. We are going to prove that
    \begin{equation}
        \bar r_{\rm key}(\vec A)\leq \frac{1}{N-1}\sum_{e\in\mathcal E}E_{\rm F}^\infty(\rho_e).
    \end{equation}

    The proof actually repeats the above prove for the bipartite case. 
    Consider arbitrary decompositions
    \begin{equation}
        \rho_e^{\otimes n}=\sum_{x^{(e)}\in\mathcal x^{(e)}}
        p_{x^{(e)}}
        \ket{\psi_{x^{(e)}}}
        \bra{\psi_{x^{(e)}}}
    \end{equation}
    for every edge and some $n>0$ and a purification
    \begin{equation}
        \ket{\Psi^{(e)}}=
        \sum_{x^{(e)}\in\mathcal X^{(e)}}
        \sqrt{p_{x^{(e)}}}
        \ket{\psi_{x^{(e)}}}
        \otimes
        \ket{x^{(e)}}.
    \end{equation}
    Here $\ket{x^{(e)}}$ are bipartite unit vectors for the corresponding edges and $\ket{x^{(e)}}$ are orthogonal vectors in Eve's space. Denote 
    \begin{gather}
        \vec x=(x^{(e)})_{e\in\mathcal E}\in\prod_{e\in\mathcal E}\mathcal X^{(e)}=\mathcal X, \\
        p_{\vec x}=
        \prod_{e\in\mathcal E}
        p_{x^{(e)}},
        \quad
        \ket{\psi_{\vec x}}
        =\bigotimes_{e\in\mathcal E}
        \ket{\psi_{x^{(e)}}},
        \quad
        \ket{\vec x}
        =\bigotimes_{e\in\mathcal E}
        \ket{x^{(e)}}.
    \end{gather}
    Now for $M\geq1$, we again introduce notations $\mathbf x=(\vec x_1,\ldots,\vec x_M)\in\mathcal X^M$, and 
    \begin{equation}
        p_{\mathbf x}=\prod_{i=1}^M p_{\vec x_{i}},
        \quad
        \ket{\psi_{\mathbf x}}
        =\bigotimes_{i=1}^M
        \ket{\psi_{\vec x_{i}}},
        \quad
        \ket{\mathbf x}
        =\bigotimes_{i=1}^M
        \ket{\vec x_{i}}.
    \end{equation}
    That is, like in the proof for the bipartite case, we again consider sequences $\mathbf x$, but now we simply consider $\vec x$ as their ``letters''. In this notation, Eq.~(\ref{EqPsiABEnM}) is again true and we can repeat all the reasonings. 
    The only difference is as follows. Instead of Ineq.~(\ref{EqMeanSbnd}), we prove
    \begin{equation}
    \begin{split}
    \label{EqMeanSentMulti}
        \overline R&
        \leq
        \sum_{\vec x\in\mathcal X}
        \frac{p_{\vec x}}{N-1}
        \sum_{e\in\mathcal E}
        S_{\rm ent}(\ket{\psi_{x^{(e)}}}).        
        \\&=
        \frac{1}{N-1}
        \sum_{e\in\mathcal E}
        \sum_{
        x^{(e)}\in\mathcal X^{(e)}
        }
        p_{x^{(e)}}
        S_{\rm ent}(\ket{\psi_{x^{(e)}}})
        \\
        &=
        \frac{1}{N-1}
        \sum_{e\in\mathcal E}
        \overline S_e,
    \end{split}
    \end{equation}
    where we have introduced the notation $\overline S_e$ for the sum over $x^{(e)}$.
    Correspondingly, we need to modify Eq.~(\ref{EqPsitSent}). We define now the state $\ket{\psi_T}$ analogously:
    \begin{equation}
        \ket{\psi_T}
        =
        \bigotimes_{\vec x\in\mathcal X}
        \ket{\psi_{\vec x}}
        ^{\otimes
        \lceil
        M_0(p_{\vec x}+\delta_{M_0})\rceil
        }.
    \end{equation}
    According to Theorem~\ref{ThPENpure}, its distillable conference key is upper bounded (actually, is equal, as we know from the main text) by
    \begin{equation}
        \sum_{\vec x\in\mathcal X}
        \frac{
            \lceil 
            M_0(p_{\vec x}+\delta_{M_0})
        \rceil
        }{N-1}
        \sum_{e\in\mathcal E}
        S_{\rm ent}(\ket{\psi_{x^{(e)}}}).
    \end{equation}
    That is, we have replaced $S_{\rm ent}(\ket{\psi_{x}})$ with $\frac1{N-1}\sum_e S_{\rm ent}(\ket{\psi_{x^{(e)}}})$. All other reasonings are repeated without modification (up to replacement of $x$ and $\overline S_{\rm ent}$ with $\frac{1}{N-1}\sum_e\overline S_e$). This finished the proof for the case $I=[N]$ and the finest partition $\mathcal P$:
    \begin{equation}
    \label{EqKeyRateEFfinest}
        \overline r_{\rm key}
        (\vec A)
        \leq
        \frac{1}{N-1}
        \sum_{e\in\mathcal E}
        E_{\rm F}^\infty(\rho_e).
    \end{equation}

    The general case of an arbitrary $I\subset [N]$ and an $I$-proper partition $\mathcal P$ is reduced to it. We can consider the graph contraction corresponding to the partition $\mathcal P$, i.e., the graph where the elements of $\mathcal P$ are vertices and the quantum state corresponding to the edge $(J,J')$ for $J,J'\in\mathcal P$ is
    \begin{equation}
        \rho_{(J,J')}
        =\bigotimes_
        {e\in\mathcal E(J,J')}
        \rho_e.
    \end{equation}
    Here $\mathcal E(J,J')\subset \mathcal E$ is the subset of edges of the original graph where one endvertex belongs to $J$ and the other one belongs to $J'$. Let us relax the restriction of the allowed class of operations to the $\mathcal P$-local ones (and, as before, with shared randomness and classical postprocessing), i.e., nonlocal operations inside the partition subsets are allowed. Then the problem of conference key distillation for this problem is obviously weaker then the problem of conference key distillation for the subset $I$ using the usual LOSR+PP. Thus, application of Ineq.~(\ref{EqKeyRateEFfinest}) to this setting [i.e., replacement of $N$ with $|I|$ and the summation is over $\mathcal E(\mathcal P)$] gives Ineq.~(\ref{EqKeyRateEF}).
\end{proof}

\bibliography{bibliography.bib}

@article{Neumann2025_NoAdvantageLOSR,
  author  = {Neumann, Justus and Kondra, Tulja Varun and Hansenne, Kiara and
             Weinbrenner, Lisa T. and Kampermann, Hermann and G\"uhne, Otfried and
             Bru\ss, Dagmar and Wyderka, Nikolai},
  title   = {No Quantum Advantage without Classical Communication: Fundamental Limitations of Quantum Networks},
  journal = {arXiv preprint},
  year    = {2025},
  eprint  = {2503.09473},
  eprinttype = {arXiv},
  url     = {https://arxiv.org/abs/2503.09473}
}

@article{Epping2017,
  title={Multi-partite entanglement can speed up quantum key distribution
in networks},
  author={Epping, M. and Kampermann, H. and Macchiavello, C. and Bru{\ss}, D.},
  journal={New. J. Phys.},
  volume={19},
  pages={093012},
  year={2017},
  doi={10.1088/1367-2630/aa8487}
}

@article{Carrara2021,
  title={Genuine multipartite entanglement is not a precondition for secure conference key agreement},
  author={Carrara, G. and Kampermann, H. and Bru{\ss}, D. and Murta, G.},
  journal={Phys. Rev. Res.},
  volume={3},
  number={1},
  pages={013264},
  year={2021},
  doi={10.1103/PhysRevResearch.3.013264}
}

@article{Pirandola2020,
  title={General upper bound for conferencing keys in
arbitrary quantum networks},
  author={Pirandola, S.},
  journal={IET Quantum Commun.},
  volume={1},
  number={1},
  pages={22--25},
  year={2020},
  doi={10.1049/iet-qtc.2020.0006}
}

@article{DW2005,
  title={General upper bound for conferencing keys in
arbitrary quantum networks},
  author={Devetak, I. and Winter, A.},
  journal={Proc. R. Soc. A},
  volume={207--235},
  number={4},
  pages={041016},
  year={2005},
  doi={10.1098/rspa.2004.1372}
}

@article{Augusiak2009,
  title={Multipartite secret key distillation and bound entanglement},
  author={Augusiak, R. and Horodecki, P.},
  journal={Phys. Rev. A},
  volume={80},
  number={4},
  pages={042307},
  year={2009},
  doi={10.1103/PhysRevA.80.042307}
}

@article{Vicente2022,
  title={Asymptotic survival of genuine multipartite entanglement in noisy quantum networks depends on the topology},
  author={Contreras-Tejada, P. and Palazuelos, C. and  de Vicente, J.I.},
  journal={Phys. Rev. Lett.},
  volume={128},
  number={22},
  pages={220501},
  year={2022},
  doi={10.1103/PhysRevLett.128.220501},
}

@article{Navascues2020,
  title={Genuine Network Multipartite Entanglement},
  author={Navascu{\'e}s, M. and Wolfe, E. and  Rosset, D. and Pozas-Kerstjens, A.},
  journal={Phys. Rev. Lett.},
  volume={125},
  number={24},
  pages={240505},
  year={2020},
  doi={10.1103/PhysRevLett.125.240505}
}

@article{Kraft2021,
  title={Quantum entanglement in the triangle network},
  author={Kraft, T. and Designolle, S. and Ritz, C. and  Brunner, N. and G{\"u}hne, O. and Huber, M.},
  journal={Phys. Rev. A},
  volume={103},
  number={6},
  pages={L060401},
  year={2021},
  doi={10.1103/PhysRevA.103.L060401}
}

@article{Hansenne2022,
  title={Symmetries in quantum networks lead to no-go
theorems for entanglement distribution and to
verification techniques},
  author={Hansenne, K. and Xu, Z.-P. and Kraft, T. and  G{\"u}hne, O.},
  journal={Nature Comm.},
  volume={13},
  pages={496},
  year={2022},
  doi={10.1038/s41467-022-28006-3}
}

@article{Csiszar2004,
  title={Secrecy Capacities for Multiple Terminals},
  author={Csisz{\'a}r, I. and Narayan, P.},
  journal={IEEE Trans. Inf. Theory},
  volume={50},
  issue={12},
  pages={3047--3061},
  year={2004},
  doi={10.1109/TIT.2004.838380}
}

@book{Diestel,
	author    = {Diestel, R.}, 
	title     = {Graph Theory}, 	
	address   = {Berlin},
        edition = {5th},
	publisher = {Springer},
	year      = {2017}
}

@INPROCEEDINGS{BB84,
author = "C. H. Bennett and G. Brassard", 
title="Quantum
cryptography: public key distribution and coin tossing",
year="1984",
booktitle="Proc. IEEE Int. Conf. Computers, Systems
and Signal Processing",
publisher="Institute of Electrical and
Electronics Engineers",
address="New York",
pages = "175--179"
}

@ARTICLE{Kimble,
author = "H. Kimble", 
title="The quantum internet",
year = "2008",
journal = "Nature",
volume= "453",
issue="",
pages = "1023--1030",
doi = "10.1038/nature07127"
}

@ARTICLE{QIntenet-Pirandola,
author = "Pirandola, S. and Braunstein, S.", 
title="Physics: Unite to build a quantum Internet",
year = "2016",
journal = "Nature",
volume= "532",
issue="",
pages = "169--171",
doi = "10.1038/532169a"
}

@ARTICLE{QIntenet-Simon,
author = "C. Simon", 
title="Towards a global quantum network",
year = "2017",
journal = "Nature Photon.",
volume= "11",
issue="",
pages = "678--680",
doi = "10.1038/s41566-017-0032-0"
}

@ARTICLE{QIntenet-Wehner,
author = "Stephanie Wehner and David Elkouss and Ronald Hanson", 
title="Quantum internet: A vision
for the road ahead",
year = "2018",
journal = "Nature Photon.",
volume= "11",
issue="6412",
pages = "eaam9288",
doi = "10.1126/science.aam9288"
}

@book{QInternet-book,
	author    = {P. P. Rohde}, 
	title     = {The Quantum Internet. The second Quantum Revolution}, 	
	address   = {Cambridge},
        edition = {},
	publisher = {Cambridge University Press},
	year      = {2021}
}

@misc{WildeBook,
      title={Principles of Quantum Communication Theory: A Modern Approach}, 
      author={Sumeet Khatri and Mark M. Wilde},
      year={2024},
      eprint={2011.04672},
      archivePrefix={arXiv},
      primaryClass={quant-ph},
      url={https://arxiv.org/abs/2011.04672}, 
}

@misc{ChristandlThesis,
      title={The Structure of Bipartite Quantum States - Insights from Group Theory and Cryptography}, 
      author={Matthias Christandl},
      year={2006},
      eprint={quant-ph/0604183},
      archivePrefix={arXiv},
      primaryClass={quant-ph},
      url={https://arxiv.org/abs/quant-ph/0604183}, 
}

@book{HolevoBook,
	author    = {Holevo, A.S.}, 
	title     = {Quantum Systems, Channels, Information.
A Mathematical Introduction}, 	
	address   = {Berlin, Boston},
        edition = {},
	publisher = {De Gruyter},
	year      = {2013}
}

@book{CsiszarKornerBook,
	author    = {I. Csisz\'ar and J. K\"orner}, 
	title     = {Information Theory: Coding Theorems for Discrete Memoryless Systems}, 	
	address   = {Cambridge},
        edition = {2nd},
	publisher = {Cambridge University Press},
	year      = {2011}
}

@article{AudenaertEisert2011,
    author = {Audenaert, Koenraad M. R. and Eisert, Jens},
    title = {Continuity bounds on the quantum relative entropy -- {II}},
    journal = {J. Math. Phys.},
    volume = {52},
    number = {11},
    pages = {112201},
    year = {2011},
    month = {11},
    issn = {0022-2488},
    doi = {10.1063/1.3657929},
    url = {https://doi.org/10.1063/1.3657929},
}

@article{VedralPlenio98,
  title = {Entanglement measures and purification procedures},
  author = {Vedral, V. and Plenio, M. B.},
  journal = {Phys. Rev. A},
  volume = {57},
  issue = {3},
  pages = {1619--1633},
  numpages = {0},
  year = {1998},
  month = {Mar},
  publisher = {American Physical Society},
  doi = {10.1103/PhysRevA.57.1619},
  url = {https://link.aps.org/doi/10.1103/PhysRevA.57.1619}
}

@article{DonaldHoro99,
title = {Continuity of relative entropy of entanglement},
journal = {Phys. Lett. A},
volume = {264},
number = {4},
pages = {257-260},
year = {1999},
issn = {0375-9601},
doi = {https://doi.org/10.1016/S0375-9601(99)00813-0},
url = {https://www.sciencedirect.com/science/article/pii/S0375960199008130},
author = {Matthew J. Donald and Micha{\l} Horodecki},
}

@article{MultiSqEntWilde,
  author={Kaushik P. Seshadreesan and Masahiro Takeoka and Mark M. Wilde},
  title={Bounds on Entanglement Distillation and Secret Key Agreement for Quantum Broadcast Channels},
  journal={IEEE Trans. Inf. Theory},
  volume={62},
  number={5},
  pages={2849--2866},
  year={2016},
  doi={10.1109/TIT.2016.2544803}
}

@article{MultiSqEntHoro,
  author={Dong Yang and Karol Horodecki and Michal Horodecki and Pawel Horodecki and Jonathan Oppenheim and Wei Song},
  title={Squashed Entanglement for Multipartite States and Entanglement Measures Based on the Mixed Convex Roof},
  journal={IEEE Trans. Inf. Theory},
  volume={55},
  number={7},
  pages={3375--3387},
  year={2009},
  doi={10.1109/TIT.2009.2021373}
}

@misc{QInternet2025,
      title={The Quantum Internet (Technical Version)}, 
      author={Peter P. Rohde and Zixin Huang and Yingkai Ouyang and He-Liang Huang and Zu-En Su and Simon Devitt and Rohit Ramakrishnan and Atul Mantri and Si-Hui Tan and Nana Liu and Scott Harrison and Chandrashekar Radhakrishnan and Gavin K. Brennen and Ben Q. Baragiola and Jonathan P. Dowling and Tim Byrnes and William J. Munro},
      year={2025},
      eprint={2501.12107},
      archivePrefix={arXiv},
      primaryClass={quant-ph},
      url={https://arxiv.org/abs/2501.12107}, 
}

@misc{QInternet2025Tech,
      title={Quantum Internet: Technologies, Protocols, and Research Challenges}, 
      author={Vinay Kumar and Claudio Cicconetti and Marco Conti and Andrea Passarella},
      year={2025},
      eprint={2502.01653},
      archivePrefix={arXiv},
      primaryClass={quant-ph},
      url={https://arxiv.org/abs/2502.01653}, 
}

@misc{QPhotNet2025,
      title={A Large-Scale Reconfigurable Multiplexed Quantum Photonic Network}, 
      author={Natalia Herrera Valencia and Annameng Ma and Suraj Goel and Saroch Leedumrongwatthanakun and Francesco Graffitti and Alessandro Fedrizzi and Will McCutcheon and Mehul Malik},
      year={2025},
      eprint={2501.07272},
      archivePrefix={arXiv},
      primaryClass={quant-ph},
      url={https://arxiv.org/abs/2501.07272}, 
}

@misc{TrushSpanTrees,
      title={Spanning-tree-packing protocol for conference key propagation in quantum networks}, 
      author={Anton Trushechkin and Hermann Kampermann and Dagmar Bru{\ss}},
      year={2025},
      eprint={},
      archivePrefix={},
      primaryClass={quant-ph},
      url={}, 
}

@article{Shi2013MultiPartyQKA,
  author  = {Shi, Runhua and Zhong, Hong},
  title   = {Multi-party quantum key agreement with Bell states and Bell measurements},
  journal = {Quantum Information Processing},
  year    = {2013},
  volume  = {12},
  number  = {2},
  pages   = {921--932},
  doi     = {10.1007/s11128-012-0443-2},
  url     = {https://doi.org/10.1007/s11128-012-0443-2}
}

@article{Shukla2014BellQKA,
  author  = {Shukla, Chitra and Alam, Nasir and Pathak, Anirban},
  title   = {Protocols of quantum key agreement solely using Bell states and Bell measurement},
  journal = {Quantum Information Processing},
  year    = {2014},
  volume  = {13},
  number  = {11},
  pages   = {2391--2405},
  doi     = {10.1007/s11128-014-0784-0},
  url     = {https://doi.org/10.1007/s11128-014-0784-0}
}

@article{Murta2020Review,
  author  = {Murta, Gl\'aucia and Grasselli, Federico and Kampermann, Hermann and Bru\ss, Dagmar},
  title   = {Quantum Conference Key Agreement: A Review},
  journal = {Advanced Quantum Technologies},
  year    = {2020},
  volume  = {3},
  number  = {11},
  pages   = {2000025},
  doi     = {10.1002/qute.202000025},
  url     = {https://doi.org/10.1002/qute.202000025}
}

@article{Pickston2023NetworkQCKA,
  author  = {Pickston, Alexander and Ho, Joseph and Ulibarrena, Andr{\'e}s and Grasselli, Federico and Proietti, Massimiliano and Morrison, Christopher~L. and Barrow, Peter and Graffitti, Francesco and Fedrizzi, Alessandro},
  title   = {Conference key agreement in a quantum network},
  journal = {npj Quantum Information},
  year    = {2023},
  volume  = {9},
  pages   = {82},
  doi     = {10.1038/s41534-023-00750-4},
  url     = {https://doi.org/10.1038/s41534-023-00750-4}
}

@article{Das2021,
  author  = {Das, Siddhartha and B{\"a}uml, Stefan and Winczewski, Marek and Horodecki, Karol},
  title   = {Universal Limitations on Quantum Key Distribution over a Network},
  journal = {Physical Review X},
  year    = {2021},
  volume  = {11},
  number  = {4},
  pages   = {041016},
  doi     = {10.1103/PhysRevX.11.041016},
  url     = {https://doi.org/10.1103/PhysRevX.11.041016}
}

@article{Devetak2005DWRate,
  author  = {Devetak, Igor and Winter, Andreas},
  title   = {Distillation of Secret Key and Entanglement from Quantum States},
  journal = {Proceedings of the Royal Society A},
  year    = {2005},
  volume  = {461},
  number  = {2053},
  pages   = {207--235},
  doi     = {10.1098/rspa.2004.1372},
  url     = {https://doi.org/10.1098/rspa.2004.1372}
}

@article{augusiak2009multipartite,
  title={Multipartite secret key distillation and bound entanglement},
  author={Augusiak, Remigiusz and Horodecki, Pawe{\l}},
  journal={Physical Review A—Atomic, Molecular, and Optical Physics},
  volume={80},
  number={4},
  pages={042307},
  year={2009},
  publisher={APS}
}

@article{Navascues2020_GNME,
  author  = {Navascu{\'e}s, Miguel and Wolfe, Elie and Rosset, Denis and Pozas-Kerstjens, Alejandro},
  title   = {Genuine Network Multipartite Entanglement},
  journal = {Physical Review Letters},
  year    = {2020},
  volume  = {125},
  pages   = {240505},
  doi     = {10.1103/PhysRevLett.125.240505},
  url     = {https://doi.org/10.1103/PhysRevLett.125.240505}
}

@article{Grasselli2021AnonymousCKA,
  author  = {Grasselli, Federico and Murta, Gláucia and de Jong, Jarn and Hahn, Frederik and Bruß, Dagmar and Kampermann, Hermann and Pappa, Anna},
  title   = {Secure Anonymous Conferencing in Quantum Networks},
  journal = {arXiv preprint},
  year    = {2021},
  eprint  = {2111.05363},
  eprinttype = {arXiv},
  url     = {https://arxiv.org/abs/2111.05363}
}

@article{inflation2,
  title = {Quantum Inflation: A General Approach to Quantum Causal Compatibility},
  author = {Wolfe, Elie and Pozas-Kerstjens, Alejandro and Grinberg, Matan and Rosset, Denis and Ac\'{\i}n, Antonio and Navascu\'es, Miguel},
  journal = {Phys. Rev. X},
  volume = {11},
  issue = {2},
  pages = {021043},
  numpages = {24},
  year = {2021},
  month = {May},
  publisher = {American Physical Society},
  doi = {10.1103/PhysRevX.11.021043},
  url = {https://link.aps.org/doi/10.1103/PhysRevX.11.021043}
}

@article{infaltion1,
  author       = {Elie Wolfe and Robert W. Spekkens and Tobias Fritz},
  title        = {The Inflation Technique for Causal Inference with Latent Variables},
  journal      = {Journal of Causal Inference},
  year         = {2019},
  volume       = {7},
  number       = {2},
  pages        = {20170020},
  doi          = {10.1515/jci-2017-0020},
  url          = {https://doi.org/10.1515/jci-2017-0020},
  note         = {Article ID 20170020}
}

@article{chaves2015,
  title={Information--theoretic implications of quantum causal structures},
  author={Chaves, Rafael and Majenz, Christian and Gross, David},
  journal={Nature communications},
  volume={6},
  number={1},
  pages={5766},
  year={2015},
  publisher={Nature Publishing Group UK London}
}

@article{Buscemi2004,
  title = {All Entangled Quantum States Are Nonlocal},
  author = {Buscemi, Francesco},
  journal = {Phys. Rev. Lett.},
  volume = {108},
  issue = {20},
  pages = {200401},
  numpages = {5},
  year = {2012},
  month = {May},
  publisher = {American Physical Society},
  doi = {10.1103/PhysRevLett.108.200401},
  url = {https://link.aps.org/doi/10.1103/PhysRevLett.108.200401}
}

@article{guehne2005,
  title = {Entanglement detection in the stabilizer formalism},
  author = {T\'oth, G\'eza and G\"uhne, Otfried},
  journal = {Phys. Rev. A},
  volume = {72},
  issue = {2},
  pages = {022340},
  numpages = {14},
  year = {2005},
  month = {Aug},
  publisher = {American Physical Society},
  doi = {10.1103/PhysRevA.72.022340},
  url = {https://link.aps.org/doi/10.1103/PhysRevA.72.022340}
}

@article{PerfectOmni,
author = {Nitinawarat, Sirin and Narayan, Prakash},
title = {Perfect omniscience, perfect secrecy, and Steiner tree packing},
year = {2010},
issue_date = {December 2010},
publisher = {IEEE Press},
volume = {56},
number = {12},
issn = {0018-9448},
journal = {IEEE Trans. Inf. Theor.},
month = dec,
pages = {6490–6500},
numpages = {11},
url = {https://doi.org/10.1109/TIT.2010.2081450},
doi= {10.1109/TIT.2010.2081450}
}

@article{Wooltorton,
  title = {Genuine Multipartite Entanglement is Not Necessary for Standard Device-Independent Conference Key Agreement},
  author = {Wooltorton, Lewis and Brown, Peter and Colbeck, Roger},
  journal = {Phys. Rev. Lett.},
  volume = {135},
  issue = {22},
  pages = {220803},
  numpages = {7},
  year = {2025},
  month = {Nov},
  publisher = {American Physical Society},
  doi = {10.1103/v4s8-3zl5},
  url = {https://link.aps.org/doi/10.1103/v4s8-3zl5}
}

@ARTICLE{PIN2010,
  author={Nitinawarat, Sirin and Ye, Chunxuan and Barg, Alexander and Narayan, Prakash and Reznik, Alex},
  journal={IEEE Trans. Inf. Theory}, 
  title={Secret Key Generation for a Pairwise Independent Network Model}, 
  year={2010},
  volume={56},
  number={12},
  pages={6482-6489},
  doi={10.1109/TIT.2010.2081210}
}

@article{Grasselli2018,
doi = {10.1088/1367-2630/aaec34},
url = {https://doi.org/10.1088/1367-2630/aaec34},
year = {2018},
month = {nov},
publisher = {IOP Publishing},
volume = {20},
number = {11},
pages = {113014},
author = {Grasselli, Federico and Kampermann, Hermann and Bruß, Dagmar},
title = {Finite-key effects in multipartite quantum key distribution protocols},
journal = {New Journal of Physics}
}

@article{KoashiWinter2004,
  title = {Monogamy of quantum entanglement and other correlations},
  author = {Koashi, Masato and Winter, Andreas},
  journal = {Phys. Rev. A},
  volume = {69},
  issue = {2},
  pages = {022309},
  numpages = {6},
  year = {2004},
  month = {Feb},
  publisher = {American Physical Society},
  doi = {10.1103/PhysRevA.69.022309},
  url = {https://link.aps.org/doi/10.1103/PhysRevA.69.022309}
}

\end{document}